\documentclass{article}
               \usepackage{listings}
               \usepackage{color}
               \usepackage{tocloft}

               \usepackage{times}
\usepackage{bm}
\usepackage{natbib}

\usepackage{hyphenat}
\usepackage{bbold}
\usepackage{bbm}
\usepackage[T1]{fontenc}
\usepackage{amsmath}
\usepackage{amssymb}
\usepackage{prodint}
\usepackage{color}
\usepackage{subcaption}

\usepackage{amsmath}
\usepackage{array}
\usepackage[T1]{fontenc} 
\usepackage{natbib}
\setcitestyle{authoryear,open={(},close={)}}
\usepackage{authblk} 
\usepackage{mathrsfs}
\usepackage{enumitem}

\newcommand*{\gammapoor}{\ensuremath{\gamma_{a=0}}} 
\usepackage[small]{titlesec}
\titleformat{\subsubsection}
{\normalfont\bfseries}{\thesubsubsection}{1em}{}
\usepackage{adjustbox}
\usepackage{prodint}
\usepackage{booktabs}
\usepackage{bm}
\usepackage{bbold}
\usepackage{wasysym}
\usepackage[margin=1.3in]{geometry}
\newenvironment{proof}{\textit{Proof.}}{\hfill$\square$}

\newcommand{\EE}{\mathbb{E}}

\newcommand{\eps}{\varepsilon}

\newcommand{\R}{\mathbb{R}}

\newcommand{\1}{\mathbb{1}}

\newcommand{\tmax}{\tau}

\newcommand\independent{\protect\mathpalette{\protect\independenT}{\perp}}\def\independenT#1#2{\mathrel{\rlap{$#1#2$}\mkern2mu{#1#2}}}
\newtheorem{thm}{Theorem}[section]

\newtheorem{lemma}{Lemma}

\newtheorem{example}{Example}

\usepackage{tikz}
\usetikzlibrary{arrows}
\tikzset{every picture/.style=remember picture}

\usepackage[utf8]{inputenc}
\usepackage[T1]{fontenc}
\usepackage{graphicx}
\usepackage{grffile}
\usepackage{longtable}
\usepackage{wrapfig}
\usepackage{rotating}
\usepackage[normalem]{ulem}
\usepackage{amsmath}
\usepackage{textcomp}
\usepackage{amssymb}
\usepackage{capt-of}
\usepackage{hyperref}
\numberwithin{equation}{section}

\title{Nonparametric estimation of the interventional disparity indirect effect among the exposed}
\author{Helene C. W. Rytgaard \& Amalie Lykkemark Møller \& Thomas A. Gerds}

\begin{document}

\maketitle

\begin{abstract}
  In situations with non-manipulable exposures, interventions can be
  targeted to shift the distribution of intermediate variables between
  exposure groups to define interventional disparity indirect effects.
  In this work, we present a theoretical study of identification and
  nonparametric estimation of the interventional disparity indirect
  effect among the exposed. The targeted estimand is intended for
  applications examining the outcome risk among an exposed population
  for which the risk is expected to be reduced if the distribution of
  a mediating variable was changed by a (hypothetical) policy or
  health intervention that targets the exposed population
  specifically.  We derive the nonparametric efficient influence
  function, study its double robustness properties and present a
  targeted minimum loss-based estimation (TMLE) procedure. All
  theoretical results and algorithms are provided for both uncensored
  and right-censored survival outcomes.  With offset in the ongoing
  discussion of the interpretation of non-manipulable exposures, we
  discuss relevant interpretations of the estimand under different
  sets of assumptions of no unmeasured confounding and provide a
  comparison of our estimand to other related estimands within the
  framework of interventional (disparity) effects. Small-sample
  performance and double robustness properties of our estimation
  procedure are investigated and illustrated in a simulation study.

\end{abstract}  
  
\section{Introduction}
\label{sec:introduction}

We consider the data situation
\(W \rightarrow A \rightarrow Z \rightarrow Y\) with covariates \(W\),
exposure \(A\), intermediate (mediator) variable \(Z\), and an outcome
\(Y\). We assume that \(Z\) is observed immediately after \(A\), so
that there are no exposure-induced \(Z\)-\(Y\) confounders.  For the
applications that we address, the exposure \(A\) is non-manipulable,
measuring for example health conditions, race or socioeconomic status,
and the intermediate variable \(Z\) represents an intervenable
circumstance that could be subject to policy intervention. In many
examples the exposure variable separates low from high risk
populations, and it is of interest to evaluate whether a policy
intervention on \(Z\) can reduce health disparities between the
exposure groups. In this work, our aim is to estimate the effect of a
(hypothetical) intervention which is applied to the exposed (the high
risk group), and modifies the distribution of \(Z\) to be as among the
unexposed (the low risk group). We define these effects as target
parameters in a causal inference framework and derive the efficient
influence functions in several different data settings. We describe
how to construct asymptotically linear nonparametric estimators using
the targeted minimum loss-based estimation (TMLE) procedure
\citep{van2018targeted} and theoretically verify double robustness
properties. Statistical inference for the TMLE is based on estimates
of the efficient influence function.

Our parameters of interest are defined under particular stochastic
stochastic interventions
\citep{diaz2013assessing,young2014identification,haneuse2013estimation}
shifting the observed distribution of \(Z\) to be as among the
unexposed. In mediation analysis, where the aim is to decompose total
effects into direct and indirect effects, stochastic interventions are
similarly used to define \textit{interventional} direct and indirect
effects which decompose a total exposure-outcome effect
\citep{didelez2006direct,vanderweele2014effect,zheng2017longitudinal,vansteelandt2017interventional,diaz2020causalstochastic,nguyen2021clarifying}.
In contrast to direct effects and indirect effects defined in terms of
unobservable composite counterfactuals
\citep{robins1992identifiability,pearl2001direct}, the identification
of interventional direct and indirect effects does not rely on
cross-world assumptions. Furthermore, interventional direct and
indirect effects are often argued to be policy-relevant
\citep{vanderweele2013policy}, as they actually correspond to
conceivable interventions on intermediate variable distributions. Our
causal estimands are related to interventional indirect parameters,
and particularly the interventional indirect effects among the exposed
defined by \cite{vansteelandt2012natural}, but they are not
identical. As we discuss more thoroughly in Sections
\ref{sec:interventional:interpret:1} and \ref{sec:related:parameters},
the interpretation differ to the extend that we do not consider
interventions on the exposure variables.

Our work is intended for studies of non-manipulable exposures for
which real-life interventions are not meaningful
\citep{vanderweele2012causal,vanderweele2014causal}. In continuation
of related considerations by
\cite{micali2018maternal,naimi2016mediation}, we refer to our causal
estimands as \textit{interventional disparity indirect effects among
  the exposed}. We emphasize that, in contrast to
\cite{micali2018maternal}, we are interested only in this indirect
effect in the subpopulation which is actually exposed. As is often the
case (also outside the mediation context), the exposed population can
be highly different than the unexposed, and policy makers often deal
with the decision of whether to implement interventions among the
exposed population specifically, or more broadly. Similarly, the
treatment effect among the treated can be of more interest than
average effects when, for example, treated individuals are very
different from the untreated
\citep{imbens2004nonparametric,heckman2001policy}.

We derive the efficient influence functions for the statistical
parameters representing our causal estimands under structural
assumptions both in the setting with uncensored outcome and in the
case where the outcome is a right censored time to event variable. We
then propose a targeted minimum loss-based estimation procedure
\citep{van2011targeted,van2018targeted}. For all settings, we study
the double (multiple) robustness properties by computing the relevant
second-order remainders, and we show that these remainders display a
desired structure that enables utilization of machine learning based
initial estimation achieving rates of \(n^{-1/4}\), such as the highly
adaptive lasso
\citep{benkeser2016highly,van2017generally,rytgaard2021estimation}. The
fact that the effect is targeted to the exposed, rather than to the
full population, changes the statistical estimation problem and the
statistical estimation procedure. We highlight these differences and
further illustrate them in our simulation study.

The article is organized as follows. Section
\ref{sec:setting:notation} introduces notation for our setting with
uncensored outcome variables, as well as two running examples.
Section \ref{sec:target:parameter} introduces the target parameter,
discusses the causal assumptions needed for different interventional
interpretations, and relates the parameter to other related parameters
in the literature. Section \ref{sec:statistical:estimation:problem}
analyzes the statistical estimation problem, presenting the efficient
influence function and results on the double robustness properties of
the estimation problem. Section \ref{sec:tmle} introduces our targeted
minimum loss-based estimation procedure. Section
\ref{sec:simulation:study} presents a simulation study to investigate
small-sample properties, verify double robustness properties and
illustrate the implications of unmeasured \(A\)-\(Z\) confounding.
The extension of our methods to settings where the outcome is a
right-censored time to event variable is given in Section
\ref{sec:event:history:setting}.  Section \ref{sec:discussion} closes
with a discussion.

\section{Setting and notation}
\label{sec:setting:notation}

We consider a setting with subject-specific observed data on the form
\(O = (W, A, Z, Y)\), where \(W\in\R^d\) are covariates,
\(A\in\lbrace 0,1\rbrace\) is a binary exposure variable,
\(Z\in \lbrace 0,1\rbrace\) is a binary intermediate variable, and
\(Y\in\lbrace 0,1\rbrace\) is a binary outcome variable. In the
applications that we have in mind, the exposure \(A\) defines two
subpopulations, an unexposed (\(A=0\)) population and an exposed
(\(A=1\)) population. The general population may contain other
subpopulations. Examples \ref{ex:motivating:1} and
\ref{ex:motivating:2} below describe two different motivating
applications from previous research
\citep{moller2022hypothetical,andersen2021mediating}.

\begin{example}
  Consider data on emergency calls, where individuals calling the
  emergency medical services report chest pain (\(A=0\)) or not
  (\(A=1\)), and subsequently are either dispatched an ambulance
  (\(Z=1\)) or not (\(Z=0\)).  Survival status (for now uncensored)
  after 30 days is measured by \(Y\in \lbrace 0,1\rbrace\), with
  \(Y=0\) if the individual is alive and \(Y=1\) if not. In this
  example, the symptom presentation of chest pain (or absence of chest
  pain) is considered the non-manipulable exposure and the ambulance
  dispatch is the manipulable intermediate variable, the distribution
  of which could be changed if for example the dispatch protocols were
  modified. We are interested in learning the expected change in
  survival among the individuals presenting without chest pain, had
  they been as likely to receive emergency ambulances as similar
  individuals who presented with chest pain.
  \label{ex:motivating:1} 
\end{example}    

\begin{example}
  Consider data on low and high income heart failure patients where
  \(A=1\) represents low income, \(A=0\) represents high income, \(Z\)
  is an indicator of initiating medical treatment after the heart
  failure diagnosis, and \(Y\in\lbrace 0,1\rbrace\) represents
  survival status after 1 year, with \(Y=0\) if the individual is
  alive and \(Y=1\) if not.  The indicator of low income is considered
  the non-manipulable exposure and the variable indicating treatment
  initiation is the manipulable intermediate variable. In this
  example, we are interested in the expected change in survival among
  the low-income heart failure patients, had they been as likely to
  initiate medical treatment as similar high income heart failure
  patients.
  \label{ex:motivating:2}
\end{example}

We assume that we observe data of \(n\) independent subjects,
\(O_1, \ldots, O_n \overset{iid}{\sim} P_0\), where \(P_0\) belongs to
a nonparametric statistical model \(\mathcal{M}\). Throughout,
corresponding to a \(P\in\mathcal{M}\), we let
\(\pi (a\mid w)=P(A=a\mid W=w)\) denote the conditional distribution
of exposure \(A\) given covariates and \(\bar{\pi}(a) = P(A=a)\)
denote the marginal distribution of \(A\).  We let
\(\gamma(z \mid a, w)=P(Z=z\mid A=a,W=w)\) denote the conditional
distribution of the intermediate variable \(Z\) given exposure and
covariates, and we use \(Q ( z, a, w) = \EE[Y \mid Z=z, A=a, W=w]\) to
denote the conditional expectation of \(Y\) given the intermediate
variable, the exposure and the covariates. In Section
\ref{sec:event:history:setting}, we extend the setting to cover
right-censored event time outcome in the presence of competing risks.

\section{Target parameter}
\label{sec:target:parameter}

We are interested in quantifying the effect of a health or policy
intervention that is implemented in the exposed subpopulation,
targeting modification of the intermediate variable \(Z\) by shifting
the distribution of \(Z\) to be as among the unexposed.  This
corresponds to a particular stochastic intervention on \(Z\), and we
define our statistical target parameter
\(\Psi \,:\,\mathcal{M}\rightarrow\R\) as follows
\begin{align}
  \Psi(P) =
  \EE \bigg[ \sum_{z=0,1} \EE[ Y \mid A, Z=z, W]  \big( \gamma(z \mid A=0, W) - \gamma(z \mid A, W)\big) \bigg\vert A=1\bigg].
  \label{eq:statistical:target:parameter}
\end{align}
The parameter \(\Psi(P)\) represents the average outcome difference
among the exposed (\(A=1\)) when changing the probability distribution
of the intermediate variable \(Z\) to be as among the unexposed
(\(A=0\)).

\subsection{Interventional interpretations formulated in a
  counterfactual framework}
\label{sec:interventional:interpret:1}

To formally discuss the interpretation of the parameter defined in
Equation \eqref{eq:statistical:target:parameter}, we use a framework
of counterfactual variables
\citep{neyman1923applications,rubin1974estimating,robins1986new,robins1987addendum}. First
define \(Y^z\) as the counterfactual version of \(Y\) that would be
observed were \(Z\) intervened upon and set to \(z\). Recall next that
\(\gamma(z \mid a, w)\) denotes the conditional distribution of the
intermediate variable \(Z\) given exposure and covariates. We now
further denote by
\begin{align}
\gammapoor(z \mid w) = P( Z=z \mid A=0, W=w), \quad z=0,1,
  \label{eq:gamma:0}
\end{align}
the distribution of \(Z\) in the unexposed group, conditional on
covariates \(W\), and by
\begin{align}
  \gamma_{a=1}(z \mid w) = P( Z=z \mid A=1, W=w), \quad z=0,1,
  \label{eq:gamma:1}
\end{align}
the conditional distribution of \(Z\) in the exposed group. Note that
the distributions defined by \eqref{eq:gamma:0}--\eqref{eq:gamma:1} do
not correspond to counterfactual distributions but to distributions in
the subpopulations defined by the exposure \(A\).  Let
\(Z^{\gammapoor}\) denote a random variable with distribution
\(\gammapoor\). With this notation we have \(Z^{\gamma} =
Z\). Correspondingly, denote by \(Y^{{\gammapoor}}\) the
counterfactual outcome we would observe had we intervened and changed
the conditional distribution of \(Z\) to be \(\gammapoor\) rather than
\(\gamma\). Similarly define \(Z^{\gamma_{a=1}}\) and
\(Y^{{\gamma_{a=1}}}\). We state the following structural assumptions
that we need for the interventional interpretation of our target
parameter stated in Lemma \ref{interpret:causal:1} below:
\begin{enumerate}
\item[(A1)] \(Y^z \independent Z \mid (A, W)\), for \(z=0,1\);
\item[(A2)] \(Y^{z} = Y\) if \(Z=z\), for \(z=0,1\); 
\item[(A3)] \(P(Z=z \mid A=1, W)>\eta_1>0\) for \(z=0,1\) on the support
  of the distribution of \(W\) in the exposed;
\item[(A4)] \(P(A=a \mid W)>\eta_2>0\) for \(a=0,1\) on the support of
  the distribution of \(W\) in the exposed.
\end{enumerate}

In Figure \ref{fig:simple:dag:0}, Assumption (A1) corresponds to the
absence of unobserved variables with direct arrows into \(Z\) and
\(Y\).

\begin{lemma}[Interventional interpretation]
  The parameter defined by \eqref{eq:statistical:target:parameter}
  identifies the causal parameter
\begin{align}
  \Psi(P) =  \EE [ Y^{{\gammapoor}} - Y^{{\gamma_{a=1}}}
  \mid A=1], 
  \label{eq:target:parameter:1}
\end{align}
under Assumptions (A1)--(A4). 
\label{interpret:causal:1}\end{lemma}

\begin{proof}
See Appendix \ref{app:lemma:1:proof}. 
\end{proof}       \\
   
\begin{figure}[h] 
  \centering \includegraphics[width=0.45\textwidth,angle=0]
{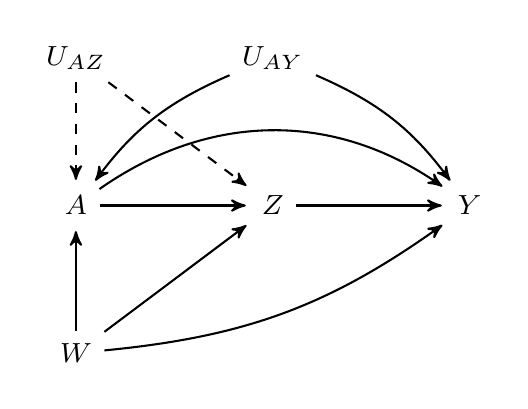}
\caption{Simple causal diagram to show the direction of relations
  between observed variables and the assumptions on unmeasured
  confounding. Moving from Lemma \ref{interpret:causal:1} to Lemma
  \ref{interpret:causal:2} requires the assumption that there are no
  direct (here dashed) arrows from \(U_{AZ}\) into both \(A\) and
  \(Z\). }
\label{fig:simple:dag:0} 
\end{figure}

Under additional structural assumptions we can gain a stronger
interpretation of the target parameter defined by
\eqref{eq:statistical:target:parameter}. To characterize the stronger
interpretation, we denote by \(Z^0\) the counterfactual intermediate
variable that we would observe had we intervened and set \(A=0\), and
by \(Z^1\) as the counterfactual intermediate variable we would
observe had we intervened and set \(A=1\). We denote the conditional
distributions of \(Z^0\) and \(Z^1\) conditional on covariates \(W\)
by \(\gamma^{a=0}\) and \(\gamma^{a=1}\), respectively, and by
\(Y^{\gamma^{a=0}}\) and \(Y^{\gamma^{a=1}}\) the counterfactual
outcomes we would have seen had \(Z\) followed the distributions
\(\gamma^{a=0}\) and \(\gamma^{a=1}\), respectively. We emphasize the
difference between \(\gamma_{a=0}\) and \(\gamma^{a=0}\) (and
\(\gamma_{a=1}\) and \(\gamma^{a=1}\) correspondingly) denoting the
observed and counterfactual distribution, respectively. The following
additional structural assumptions are required to obtain Lemma
\ref{interpret:causal:2}:
\begin{enumerate} 
\item[(A1*)] \(Z^a \independent A \mid W\), for \(a=0,1\); 
\item[(A2*)] \(Z^{a} = Z\) if \(A=a\), for \(a=0,1\).
\end{enumerate}
\begin{lemma}[Stronger counterfactual interpretation]
  The parameter defined by \eqref{eq:statistical:target:parameter}
  identifies the causal parameter
\begin{align}
  \Psi(P) =  \EE [ Y^{\gamma^{a=0}} - Y^{\gamma^{a=1}} \mid A=1], 
\label{eq:target:parameter:2}
\end{align}
under Assumptions (A1)--(A4) and Assumptions (A1*)--(A2*). 
\label{interpret:causal:2}\end{lemma}

\begin{proof}
  See Appendix \ref{app:lemma:2:proof}.
\end{proof}\\

To highlight the differences between the interpretations achieved in
Lemma \ref{interpret:causal:1} and Lemma \ref{interpret:causal:2}, we
consider our motivating examples from Section
\ref{sec:setting:notation} separately. In Example
\ref{ex:motivating:1} the intermediate variable \(Z\) is an indicator
of an ambulance being sent for a person who calls the emergency
service. Here the statistical parameter
\eqref{eq:statistical:target:parameter} represents the average outcome
difference that would occur if the distribution of an ambulance being
sent was shifted to be as among the unexposed. With the interpretation
achieved in Lemma \ref{interpret:causal:1}, our target parameter
represents the effect of a real-life intervention shifting this
distribution; for example, we could image that health care workers
receiving emergency calls have some `protocol' for sending an
ambulance based on information \((A,W)\) provided by the person making
the call, and the intervention corresponds to changing this protocol.
However, it may very well be that the protocol, i.e.,
\( \gamma(z\mid A, W)\), depends on \(A\) only through unmeasured
confounders \(U_{AZ}\). The interpretation achieved in Lemma
\ref{interpret:causal:2}, on the other hand, allow us to say something
about the risk difference that can really be ascribed to differences
in the distributions of ambulances being sent among the unexposed and
the exposed (reporting or not reporting chest pain), unrelated to
other factors, i.e., the effect we would see if we in fact could make
subjects report chest pain or not when making the emergency call.

In Example \ref{ex:motivating:2} the intermediate variable \(Z\) is an
indicator of treatment initiation. For this example, the statistical
parameter \eqref{eq:statistical:target:parameter} represents the
difference in expected outcome arising from shifting the distribution
of treatment initiation for exposed subjects to the distribution of
treatment initiation among the unexposed.  Lemma
\ref{interpret:causal:1} allows us to translate the statistical
intervention to a real-life intervention changing the observed
treatment initiation among the low-income to what it is among the
high-income. However, it may be that \( \gamma(z\mid A, W)\), depends
on \(A\) only through unmeasured confounders \(U_{AZ}\) such as
general willingness to initiate treatment; clearly, the price of the
treatment is one important factor for the disparity between high and
low income patients, but there may be other factors, such as the
frequency of doctor visits, which may explain why patients with low
income are less likely to initiate treatment. A health policy could
reduce the price of the treatment for heart failure patients and could
even provide the treatment without costs. It is important to note that
such a health policy would only affect the dissimilarities that are
related to the price of the treatment, but not modify the other
factors. Only under the additional assumptions (A1*)--(A2*) can we
interpret the parameter as the effect we would see if a policy
intervention had been implemented to provide the treatment with lower
or without costs. The causal parameter defined by Lemma
\ref{interpret:causal:1}, on the other hand, reflects the health
policy effect on the outcome in the subgroup of low income patients
when all dissimilarity factors are removed.

\subsection{Parameters considered in related literature}
\label{sec:related:parameters}

\cite{micali2018maternal} investigate the extent to which adolescent
eating disorders are associated to maternal prepregnancy underweight
or overweight status under interventions to change the distribution of
selected childhood variables to be the same as observed among those of
children of mothers who were normal weight. In our notation, we can
represent maternal prepregnancy as the variable \(A\), collect the
selected childhood variables in the variable \(Z\) and let \(Y\)
represent adolescent eating disorder, and write the
\textit{interventional disparity indirect effect} considered by
\citep{micali2018maternal} as follows
\begin{align}
\mathrm{IDM\text{-}IE} =  \EE \big[ \EE \big[ Y^{\gamma_{a=1}} \mid A=1, W\big]
  - \EE \big[ Y^{\gammapoor} \mid A=1, W\big] \big],
  \label{eq:micali:idm-ie}
\end{align}
with the outer expectation taken over the distribution of covariates
across the entire population. The parameter defined by
\eqref{eq:micali:idm-ie} closely resembles our statistical parameter
rewritten as \eqref{eq:target:parameter:1} in Lemma
\ref{interpret:causal:1} under Assumptions (A1)--(A4), with the only
difference being the outer average taken over different distributions
of covariates. We argue that our choice of parameter represents the
one of policy relevance: we care only about assessing the impact of
imposing interventions on the exposed group particularly,
corresponding to changing the frequency of ambulance pick-up for the
subjects not reporting chest pain in Example \ref{ex:motivating:1} and
changing the willingness of initiating treatment for the low-income
patients in Example \ref{ex:motivating:2}. Since subjects reporting
and not reporting chest pain, for example, represent highly different
populations, the distribution of covariates among the two populations
are expectedly similarly different.

\cite{vansteelandt2012natural} propose another similar counterfactual
estimand, the definition of which requires a bit of extra
notation. For the purpose of presenting it here, define here the
counterfactual outcome \(Y^{a', \gamma^{a}}\), the outcome we would
observe if the exposure \(A\) was set to \(a'\) and the intermediate
(mediator) variable \(Z\) had followed the distribution it would had
taken if we changed exposure level to \(a\). In this notation,
\cite{vansteelandt2012natural} defines the \textit{natural indirect
  effect among the exposed} as follows
\begin{align}
\EE\big[ Y^{0, \gamma^{a=1}} - Y^{0, \gamma^{a=0}} \mid A=1\big]; 
  \label{eq:van:niee}
\end{align}
again, this parameter closely resembles our parameter rewritten as
\eqref{eq:target:parameter:2} in Lemma \ref{interpret:causal:2} under
Assumptions (A1)--(A4) and (A1*)--(A2*). The difference now lies in
the considered counterfactual versions of \(Y\) under interventions on
the exposure. Indeed, the parameter defined by \eqref{eq:van:niee}
considers the risk difference for the exposed group, had they in fact
been unexposed but had their intermediate variable taken values as if
they had been exposed contrasted to being unexposed. For Example
\ref{ex:motivating:2}, this would correspond to the risk difference we
would see for low-income patients had they in fact been high-income
and had they initiated treatment as if they had been low-income versus
as if they had been high-income. Our parameter defined by
\eqref{eq:target:parameter:2}, on the other hand, corresponds to the
risk difference we would see for low-income patients had they
initiated treatment as if they had been low-income versus as if they
had been high-income. The two parameters will differ to the extent
that the distribution of \(Y\) changes under direct interventions on
income level.

\section{Statistical estimation problem}
\label{sec:statistical:estimation:problem}

Our statistical target parameter defined by
\eqref{eq:statistical:target:parameter} can be represented as
\(\Psi(P) = \Psi_0 (P) - \Psi_1 (P)\), with the parameters
\(\Psi_0 \, :\, \mathcal{M}\rightarrow\R\) and
\(\Psi_1 \, :\, \mathcal{M}\rightarrow\R\) defined separately as
\begin{align}
  \Psi_0(P) =
  \EE \bigg[ \sum_{z=0,1} \EE[ Y \mid A, z, W]   \gamma(z \mid 0, W)  \bigg\vert A=1\bigg],
  \label{eq:parameter:0}
  \intertext{and,}
    \Psi_1(P) =
  \EE \bigg[ \sum_{z=0,1} \EE[ Y \mid A, z, W]   \gamma(z \mid 1, W)  \bigg\vert A=1\bigg]. 
  \label{eq:parameter:1}
\end{align}
The following theorem provides the efficient influence function for
the parameter \(\Psi_{a^*} \, : \, \mathcal{M}\rightarrow \R\),
written generally as
\begin{align*}
  \Psi_{a^*}(P)
  &= \EE \bigg[ \sum_{z=0,1} \EE[ Y \mid  Z=z, A, W]   \gamma(z \mid a^*, W)  \bigg\vert A=1\bigg] \\
  &=   \int_{\mathcal{W}} \sum_{z=0,1}  Q (z, 1, w) 
    \gamma( z \mid a^*, w) 
    \frac{\pi( 1 \mid w)}{\bar{\pi}(1)}  d\mu (w) \\
  &=   \int_{\mathcal{W}} \sum_{z=0,1}  \int_{\mathcal{Y}} y \, dP_{Y} (y  \mid z, 1, w) 
    \gamma( z \mid a^*, w) 
    \frac{\pi( 1 \mid w)}{\bar{\pi}(1)}  d\mu (w), 
\end{align*}
We also write
\(\Psi_{a^*} (P) = \tilde{\Psi}_{a^*} (\pi, \bar{\pi}, \gamma, Q)\),
and refer to \((\pi, \bar{\pi}, \gamma, Q)\) as the nuisance
parameters for the estimation problem. The efficient influence curve
characterizes the asymptotic distribution of all asymptotically linear
estimators \citep{bickel1993efficient,van2000asymptotic}, and,
particularly, constructing estimators such as to solve the efficient
influence curve equation is a necessary basis for asymptotically
linear (and efficient) estimation. We present the efficient influence
curve for \(\Psi_{a^*} \, : \, \mathcal{M}\rightarrow \R\) in Theorem
\ref{thm:eff:ic} below; Theorem \ref{thm:double:robustness} next
states the double robustness properties based on the second-order bias
term admitting a specific product structure (as shown in Appendix
\ref{app:remainder:binary}).  The product structure of the
second-order remainder term further tells us that
slower convergence rates (than the typical \(n^{1/2}\) rate) are
allowed for initial estimators, attainable, for example, with the
highly adaptive lasso estimator
\citep{benkeser2016highly,van2017generally,rytgaard2021estimation}, or
by combining multiple algorithms in a super learner
\citep{polley2011super,van2011targeted} (as long as one of the
algorithms attain the required rate). This provides the basis for
nonparametric inference, and is utilized in Theorem
\ref{thm:inference:tmle} (Section \ref{sec:inference:tmle}).

\begin{thm}[Efficient influence function]
  The efficient influence function for 
  \(\Psi_{a^*} \, : \, \mathcal{M}\rightarrow \R\) can be represented
  as follows:
\begin{align*}
  \phi_{a^*} (P) (O)
  &  = \frac{\gamma( Z \mid a^*, W) }{\gamma( Z \mid 1, W)}
    \frac{\1\lbrace A=1\rbrace}{\bar{\pi}(1)}
    \big( Y - Q(Z, A, W)\big) \\
  &\qquad\quad  + \frac{\1\lbrace A=a^*\rbrace}{\pi(A \mid W) } \frac{\pi( 1 \mid W)}{\bar{\pi}(1)}
    \bigg( Q(Z, 1, W) - \sum_{z=0,1} Q(z, 1, W) \gamma(z \mid A , W) \bigg) \\
  &\qquad\quad + \frac{\1\lbrace A=1\rbrace }{\bar{\pi}(1)} \bigg(\sum_{z=0,1} Q(z, 1, W) \gamma(z \mid a^* , W)
    - \Psi_{a^*} (P) \bigg)  .
\end{align*}
We also write
\(\phi_{a^*} (P) = \tilde{\phi}_{a^*} (\pi, \bar{\pi}, \gamma, Q)\).
\label{thm:eff:ic}
\end{thm}

\begin{proof}
  See Appendix \ref{app:eff:ic:binary}.
\end{proof}\\

We only state the double robustness properties for estimation of the
parameter \(\Psi_0(P_0)\). For estimation of \(\Psi_1(P_0)\), the
efficient influence function greatly simplifies and estimators can be
constructed straightforwardly without dependence on nuisance
parameters (see Section \ref{sec:tmle}).

\begin{thm}[Double robustness]
  Suppose given are estimators \(\hat{\pi}_n\) of  \(\pi\),   \(\hat{\gamma}_n\) of \(\gamma\) 
 and  \(\hat{Q}_n\) of \(Q\)  with large sample limits
   \({\pi}'_0\),
  \({\gamma}'_0\) and \({Q}'_0\), and define 
  \(\hat{\bar{\pi}}_n = \frac{1}{n} \sum_{i=1}^n \1\lbrace
  A_i=1\rbrace\). If 
  \(\mathbb{P}_n \tilde{\phi}_{0} (\hat{\pi}_n, \hat{\bar{\pi}}_n,
  \hat{\gamma}_n,\hat{Q}_n)=o_P(n^{-1/2})\), and
  \begin{enumerate}
\item[a.] \({\gamma}'_0 = \gamma_0\); or
\item[b.] \({\pi}'_0= \pi_0 \) and \({Q}'_0 = Q_0\),
  \end{enumerate}
  then
  \(\tilde{\Psi}_0(\hat{\pi}_n, \hat{\bar{\pi}}_n,
  \hat{\gamma}_n,\hat{Q}_n)\) is a consistent estimator for
  \(\Psi_{0} (P_0)\).
\label{thm:double:robustness}
\end{thm}

\begin{proof}
  Define
  \(R_{a^*}(P, P_0) = \Psi_{a^*}( P) - \Psi_{a^*}(P_0) + P_0
  \phi_{a^*}(P)\). We show in Appendix \ref{app:remainder:binary} that
  \begin{align*}
                      & R_{0}(P,P_0)  =\\
    & \EE_{P_0}  \bigg[ 
    \frac{\pi_0(1 \mid W) }{\bar{\pi} (1)}
    \sum_{z=0,1} \bigg(  \frac{\gamma (z \mid 1,W) - \gamma( z \mid 1, W) }{\gamma( z \mid 1, W)}  \bigg) 
    \big( Q_0(z,1,W) - Q(z, 1, W) \big)\gamma( z \mid 0, W)  \bigg)
    \bigg] \\
  & \qquad +   \EE_{P_0}  \bigg[  \bigg(  \frac{\pi_0(1 \mid W)  - \pi( 1 \mid W)
                }{\bar{\pi} (1) ( 1- \pi( 1 \mid W))} \bigg)
             \sum_{z=0,1} Q(z, 1, W) \big( \gamma(z \mid 0 , W)-  \gamma(z \mid 0 , W) \big)
  \bigg] \\
  & \qquad  + \EE_{P_0}  \bigg[\frac{\pi_0(1 \mid W) }{\bar{\pi} (1)}
    \sum_{z=0,1} 
    \big( Q_0(z,1,W) - Q(z,1,W)\big)  \big( \gamma( z \mid 0, W)  -  \gamma( z \mid 0, W) \big)
    \bigg] \\
  & \qquad + 
             \bigg( 1-     \frac{\bar{\pi}_0( 1)}{\bar{\pi}(1)} \bigg) \big( \Psi_{0} (P)
             - \Psi_{0} (P_0) \big).
  \end{align*}
  from which a. and b. follow.
\end{proof}

\section{Targeted minimum loss-based estimation (TMLE)}
\label{sec:tmle}

We present our targeting algorithms that update the nuisance parameter
estimators in order to solve the efficient influence curve equation.
As presented in Theorem \ref{thm:eff:ic}, the efficient influence
function for \(\Psi_0 \, :\, \mathcal{M}\rightarrow\R\) is given as
follows
\begin{align*} 
  \phi_0 ( P ) (O)
  & =   \frac{\gamma(Z \mid 0, W)}{
    \gamma(Z \mid 1, W)}  \frac{\1\lbrace A=1\rbrace }{\bar{\pi}(1)} \Big( 
    Y  - 
    Q(Z, A, W) \Big) \\
  & \quad + \,  \frac{\1\lbrace A=0\rbrace }{\pi (0\mid W)} \frac{\pi(1 \mid W) }{\bar{\pi}(1)}\Big(
    Q(Z, 1, W)
    -    \sum_{z=0,1}   Q(z, 1, W)
    \gamma (z \mid A, W) \Big)\\
  & \quad + \, \frac{\1\lbrace A=1\rbrace  }{\bar{\pi}(1)} \bigg(
    \sum_{z=0,1} Q( z, 1,W)
    \gamma(z \mid 0, W)   - \Psi_0 (P)\bigg) 
    , 
    \intertext{and for \(\Psi_1 \, :\, \mathcal{M}\rightarrow\R\) simply as }
    \phi_1 ( P ) (O)
  &= \frac{\1\lbrace A=1\rbrace}{\bar{\pi}(1)}
    \big( Y -  \Psi_1 (P)\big)
   . 
\end{align*}
Below we describe our targeting algorithm for solving efficient
influence curve equation for \(\Psi_0(P_0)\). No targeting will be
necessary for \(\Psi_1(P_0)\), because the plug-in estimator
\begin{align}
  \hat{\psi}^{\mathrm{tmle}}_{1,n} = \frac{1}{n} \sum_{i=1}^n \frac{\1\lbrace A_i =
  1\rbrace}{\hat{\bar{\pi}}_n(1)} Y_i ,
 \label{eq:est:tmle:1::::}
\end{align}
with
\(\hat{\bar{\pi}}_n(1) = \frac{1}{n} \sum_{i=1}^n \1\lbrace A_i =
1\rbrace\) already solves the efficient influence curve
equation.

\subsection{Targeting algorithm}
\label{sec:targeting:algorithm}

To explain our targeting algorithm, we rewrite the efficient influence
function for \(\Psi_{0} \, : \, \mathcal{M}\rightarrow \R\) on the
slightly different form:
\begin{align}
  \phi_0 ( P ) (O)
  &=   \frac{\gamma(Z \mid 0, W)}{
    \gamma(Z \mid 1, W)}  \frac{\1\lbrace A=1\rbrace }{\bar{\pi}(1)} \big( 
    Y  - 
    Q(Z, A, W) \big) \label{eq:eic:tmle:1}\\
  & \quad + \,  \frac{\1\lbrace A=0\rbrace }{\pi (0\mid W)} \frac{\pi(1 \mid W) }{\bar{\pi}(1)}
    \big( Q(1,1,W) - Q(0, 1, W)  \big) \big( Z - \gamma(1 \mid A, W)\big)
    \label{eq:eic:tmle:2}\\
  & \quad + \, \frac{
    \sum_{z=0,1} Q( z, 1, W)  
    \gamma(z \mid 0, W)  - \Psi_0 (P)  }{\bar{\pi}(1)}  \big(
    \1\lbrace A=1\rbrace - \pi(1 \mid W)\big)\label{eq:eic:tmle:3} \\
  & \quad + \, \frac{\pi(1 \mid W) }{\bar{\pi}(1)} \bigg(
    \sum_{z=0,1} Q(z, 1, W)  
    \gamma(z \mid 0, W)  - \Psi_0 (P)\bigg) 
    . \label{eq:eic:tmle:4}
\end{align}
The terms \eqref{eq:eic:tmle:1}--\eqref{eq:eic:tmle:3} will be used in
Section \ref{sec:loss:submodels} to guide the targeted update steps
for \(Q\), \(\gamma\) and \(\pi\).  The term \eqref{eq:eic:tmle:4} is
then taken care of in the construction of the final TMLE estimator by
plugging targeted estimators
\((\hat{Q}^*_n, \hat{\gamma}^*_n, \hat{\pi}^*_n)\) into
\begin{align}
  \hat{\psi}^{\mathrm{tmle}}_{0,n}  =
  \frac{1}{n} \sum_{i=1}^n \frac{\1\lbrace A_i =
  1\rbrace}{\hat{\bar{\pi}}_n(1)}  \sum_{z=0,1}
  \hat{Q}^*_n ( z, 1, W_i) 
  \hat{\gamma}^*_n (z \mid 0, W_i )
  , \label{eq:est:tmle:0}
\end{align}
where, as in \eqref{eq:est:tmle:1::::}, 
\(\hat{\bar{\pi}}_n(1) = \frac{1}{n} \sum_{i=1}^n \1\lbrace A_i =
1\rbrace\).  Assume we have at hand initial estimators
\(\hat{Q}^0_n\), \(\hat{\gamma}^0_n\) and \(\hat{\pi}^0_n\), where the
superscript `\(0\)' is used to mark that they are `initial'. The
targeting algorithm proceeds iteratively, updating these three
estimators one by one. The iterations are continued until the final
set of updated estimators
\(\hat{P}^{*}_n = (\hat{Q}^{*}_n, \hat{\gamma}^{*}_n,
\hat{\pi}^{*}_n,\hat{\bar{\pi}}_n)\) solves the efficient influence
curve equation sufficiently well:
\begin{align*}
\mathbb{P}_n \phi_0(\hat{P}^{*}_n ) =  o_P(n^{-1/2}). 
\end{align*}
In practice we continue iterations until
\(\vert \mathbb{P}_n \phi_0(\hat{P}^{*}_n )\vert \le
\hat{\sigma}_{0,n} / (\sqrt{n}\log(n)) \), where
\(\hat{\sigma}^2_{0,n} = \mathbb{P}_n (\phi_0 (\hat{P}_n))^2\)
estimates the variance of the efficient influence curve.  In Section
\ref{sec:loss:submodels} below we define the loss functions and least
favorable parametric submodels; in the subsequent Section
\ref{sec:targeting:steps} we define the \(k\)th update of the
targeting algorithm, starting from a set of current estimators
\(\hat{P}^{k}_n = (\hat{Q}^{k}_n, \hat{\gamma}^{k}_n,
\hat{\pi}^{k}_n)\), \(k \ge 0 \).

\subsubsection{Loss functions and least favorable parametric submodels}
\label{sec:loss:submodels}

We here define loss functions and least favorable submodels needed to
construct our targeting algorithm. First, denote by
\begin{align}
  H_1 (\gamma) (Z, A, W) & :=  \frac{\gamma(Z \mid 0, W)}{
                           \gamma(Z \mid 1, W)}  \frac{\1\lbrace A=1\rbrace }{\bar{\pi}(1)} ,
  \label{eq:clever:1}\\
  H_2 (Q,\pi) ( W) & :=  \frac{ \1\lbrace A=0\rbrace }{\pi (0\mid W)} \frac{\pi(1 \mid W) }{\bar{\pi}(1)}
                   \big( Q(1,1,W) - Q(0, 1, W)  \big),\label{eq:clever:2}\\
  H_3 (Q, \gamma) ( W) & :=   \frac{ \sum_{z=0,1} Q(z, 1, W)  
    \gamma(z \mid 0, W) - \Psi_0(P) }{\bar{\pi}(1)} .\label{eq:clever:3} 
\end{align}
We then define the following loss functions
\((Q, O) \mapsto \mathscr{L}_1 (Q) (O)\),
\((\gamma, O) \mapsto \mathscr{L}_2 (\gamma) (O)\),
\((\pi, O) \mapsto \mathscr{L}_3 (\pi) (O)\) and parametric submodels
\(Q_\eps, \gamma_\eps, \pi_\eps\) so that
\begin{align}
  \begin{split}
  \frac{d}{d\eps}\bigg\vert_{\eps = 0}
  \mathscr{L}_1(Q_{\eps}) (O)
  &=  H_1 (\gamma) (Z, A, W)   \big( 
    Y  - 
    Q(Z, A, W) \big),  \\
  \frac{d}{d\eps}\bigg\vert_{\eps = 0}
  \mathscr{L}_2(Q_{\eps}) (O)
  &=   H_2 (Q,\pi) ( W) \big( Z - \gamma(1 \mid A, W)\big), \\
  \frac{d}{d\eps}\bigg\vert_{\eps = 0}
  \mathscr{L}_3(\pi_{\eps}) (O)
  &=  H_3 (Q_Z) ( W)  \big(
    \1\lbrace A=1\rbrace - \pi(1 \mid W)\big),
      \end{split}\label{eq:loss:models:needed}
\end{align}
corresponding to the terms \eqref{eq:eic:tmle:1},
\eqref{eq:eic:tmle:2} and \eqref{eq:eic:tmle:3} of the efficient
influence curve equation. By straightforward calculations one may find
that the loss functions
\begin{align}
  \mathscr{L}_1(Q)(O) &= - \big( Y \log Q(Z, A, W) + (1-Y) \log (1-Q( Z, A,W))\big), \label{eq:loss:Q}  \\
  \mathscr{L}_2(\gamma)(O) &= -  \big( Z \log \gamma (1 \mid A, W) + (1-Z )
                             \log\gamma(0 \mid A, W)\big), \label{eq:loss:gamma} \\
  \mathscr{L}_3(\pi)(O) &= - \big( A  \log \pi(1 \mid W) + (1- A)
                          \log \pi(0 \mid W)\big), \label{eq:loss:pi}
\end{align}
together with the parametric submodels
\begin{align}
  \mathrm{logit}(Q_{\eps} (Z, A, W))
  &= \mathrm{logit}(Q (Z, A, W)) + \eps  H_1 (\gamma) (Z, A, W)) , \label{eq:submodel:Q}\\
  \mathrm{logit}(\gamma_{\eps}(1\mid A, W))
  &= \mathrm{logit}(\gamma(1\mid A, W)) + \eps  H_2 (Q,\pi) (W) , \label{eq:submodel:gamma}\\
    \mathrm{logit}(\pi_{\eps} (1 \mid W))
  &= \mathrm{logit}(\pi (1\mid W)) + \eps  H_3 (Q,\gamma)(W) ,\label{eq:submodel:pi}
\end{align}
fulfill the properties of \eqref{eq:loss:models:needed}.

\subsubsection{Targeting steps}
\label{sec:targeting:steps}

The targeting steps \(k\) of the targeting algorithm, to update a
current set of estimators
\(\hat{P}^{k}_n = (\hat{Q}^{k}_n, \hat{\gamma}^{k}_n,
\hat{\pi}^{k}_n)\) can now be summarized as follows:
\begin{description}
\item[Updating \(\hat{Q}^k_n\):] Estimate \(\eps\) in the submodel
  \eqref{eq:submodel:Q} by \(\hat{\eps}^{Q,k}_n\) obtained from
  running a logistic regression with outcome \(Y\), offset
  \(\mathrm{logit} (\hat{Q}^k_n (Z, A, W))\) and covariate
  \(H_1 (\hat{\gamma}^k_n) (Z, A, W)\). This corresponds to evaluating
  the submodel \eqref{eq:submodel:Q} in the estimators
  \(\hat{Q}^k_n,\hat{\gamma}^k_n\) and minimizing the loss
  \eqref{eq:loss:Q}.
\item[Updating \(\hat{\gamma}^k_n\):] Estimate \(\eps\) in the
  submodel \eqref{eq:submodel:gamma} by
  \(\hat{\eps}^{\gamma,k}_n\) obtained from running a logistic
  regression with outcome \(Z\), offset
  \(\mathrm{logit} (\hat{\gamma}^k_n (1 \mid A, W))\) and covariate
  \(H_2 (\hat{Q}^{k+1}_n,\hat{\pi}^k_n) ( W)\). This corresponds to
  evaluating the submodel \eqref{eq:submodel:gamma} in the estimators
  \(\hat{Q}^{k+1}_n,\hat{\gamma}^k_n,\hat{\pi}^k_n\) and minimizing
  the loss \eqref{eq:loss:gamma}.
\item[Updating \(\hat{\pi}^k_n\):] Estimate \(\eps\) in the submodel
  \eqref{eq:submodel:pi} by \(\hat{\eps}^{\pi,k}_n\) obtained from
  running a logistic regression with outcome \(Y\), offset
  \(\mathrm{logit} (\hat{\pi}^k_n (1 \mid W))\) and covariate
  \(H_3 (\hat{Q}^{k+1}_n,\hat{\gamma}^{k+1}_n) ( W)\). This
  corresponds to evaluating the submodel \eqref{eq:submodel:pi} in the
  estimators \(\hat{Q}^{k+1}_n,\hat{\gamma}^{k+1}_n,\hat{\pi}^k_n\)
  and minimizing the loss \eqref{eq:loss:pi}.
\end{description}
The iterations are repeated until
\(\vert \mathbb{P}_n \phi_0 (\hat{P}^{k}_n ) \vert \le
\hat{\sigma}_{0,n} / (\sqrt{n}\log n)\), where
\(\hat{\sigma}^2_{0,n} = \mathbb{P}_n (\phi_0 (\hat{P}_n))^2\)
estimates the variance of the efficient influence function.

\subsection{Inference for the targeted minimum loss-based estimator}
\label{sec:inference:tmle}

The asymptotic properties of our estimators
\(\hat{\psi}^{\mathrm{tmle}}_{0,n}, \hat{\psi}^{\mathrm{tmle}}_{1,n}\)
defined by \label{eq:est:tmle:1} and \label{eq:est:tmle:0},
respectively, is provided by the theorem below. The proof relies
directly on similar work \citet[][Theorem
A.5]{van2006targeted,van2017generally,rytgaard2021estimation,van2011targeted},
but is included in Appendix \ref{app:decomposition:proof} for completeness.

\begin{thm}[Inference for the targeted minimum loss-based estimator]
  Under Assumptions (A3) and (A4), as well as
  \begin{itemize}
  \item[(R1)] the efficient influence curve
    \((\phi_{a^*} (P)\,:\, P\in\mathcal{M})\) belongs to a Donsker
    class;
  \item[(R2)] all nuisance parameters \((\pi, \gamma, Q)\) are
    estimated at a rate faster than \(n^{-1/4}\);
    \end{itemize}
    the targeted minimum loss-based estimator
    \( \hat{\psi}^{\mathrm{tmle}}_{a^*,n} = \Psi_{a^*} (\hat{P}^*_{n})
    \) admits the representation
\begin{align} 
  \sqrt{n} \big(  \hat{\psi}^{\mathrm{tmle}}_{a^*,n} - \Psi_{a^*} (P_0)
  \big) = \sqrt{n} \, \mathbb{P}_n \phi_{a^*}
  (P_0) + o_P(1),
  \label{eq:asympt:linearity}
\end{align}
i.e., \( \hat{\psi}^{\mathrm{tmle}}_{a^*,n} \) is asymptotically
linear with influence function equal to the efficient influence
function.
\label{thm:inference:tmle}
\end{thm}

\begin{proof}
  See Appendix \ref{app:decomposition:proof}.
\end{proof}

\section{Simulation study}
\label{sec:simulation:study}

We evaluate and illustrate the performance of our estimation procedure
in three different variations over simulation settings. Overall, we
simulate variables as follows:
\begin{align}
    \label{eq:sim:1}\tag{sim-1}
  \begin{split}
  W_1 &\sim \mathrm{Ber}(0.6), \\
  W_2 &\sim U(-1,1), \\
  A & \sim \mathrm{Ber}(\mathrm{expit}(0.5-1.8W_1 +0.5 W_2^2)), \\
  Z & \sim \mathrm{Ber}(\mathrm{expit}(0.6-1.8W_1+0.5W_2^2-0.9A)), \\
    Y & \sim \mathrm{Ber}(\mathrm{expit}(-0.2-1.3W_1+W_2^2 + 0.8A  -0.6Z -1.8Z(1-W_1))).
        \end{split}
\end{align}
Here we may think of \(W_1=0\) indicating some health deficiency,
over-represented among the exposed \(A=1\) and much more likely to
have a beneficial effect of having \(Z=1\) (e.g., receiving an
ambulance dispatch). Moreover, we note that the variables are
simulated such that exposed subject are much less likely to have
\(Z=1\) (e.g., receiving an ambulance dispatch) and also have a higher
mortality (\(Y=1\)).

Since there is no unmeasured confounding, we can estimate the true
value of the target parameter by simulating a large number of times
from the counterfactual distributions:
\begin{align*}
  Z^{\gammapoor} & \sim \mathrm{Ber}(\mathrm{expit}(0.6-1.8 W_1+0.5  W_2^2-0.9 \cdot 0 )) \\
  Y^{\gammapoor} & \sim \mathrm{Ber}(\mathrm{expit}(-0.2-1.3W_1 +W_2^2 +0.8A  -0.6  Z^{\gammapoor}-1.8 Z^{\gammapoor}(1-W_1)))
                 \intertext{and, similarly, we estimate the true value of the interventional  disparity indirect effect considered by
                 \citep{micali2018maternal}, see
                 \eqref{eq:micali:idm-ie}, by simulating from the counterfactual distribution:}
                 Y_{A=1}^{\gammapoor}  & \sim \mathrm{Ber}(\mathrm{expit}(-0.2-1.3W_1 +W_2^2 +  0.9 \cdot 1   -0.6  Z^{\gammapoor}-1.8 Z^{\gammapoor}(1-W_1))). 
\end{align*}
Our parameter is here estimated to
\(\EE [ Y^{\gammapoor} - Y \mid A=1 ] \approx -0.0797\) and the
interventional disparity indirect effect considered by
\citep{micali2018maternal} is estimated to
\(\EE [ Y_{A=1}^{\gammapoor} - Y ] \approx -0.0552\), highlighting the
difference between assessing the effect of the intervention on the
full population compared to only the exposed population. Results from
a simulation study with \(n=1000\) and \(M=500\) repetitions can be
found in Table \ref{tab:sim:results:1}. Note that we used correctly
specified parametric models for all nuisance parameters in the first
row for each parameter, and then considered misspecification of each
model in turn by leaving out \(W_1\) and by including \(W_2\) rather
than \(W_2^2\).

\begin{table}[ht]
  \centering
  Estimation of our target parameter: \\
\begin{tabular}{llllll}
  \toprule
  scenario & truth & bias (initial) & bias (tmle) & SE & Cov (95\%) \\ 
  \hline\\[-0.3cm]
  all-correct & -0.0796 & -0.0001 & -0.0004 & 0.0195 & 0.96 \\ 
  miss-\(Q\) & -0.0796 & 0.0232 & -0.0003 & 0.0145 & 0.914 \\ 
  miss-\(\gamma\) & -0.0796 & 0.1042 & -0.0001 & 0.0198 & 0.96 \\ 
  miss-\(\pi\) & -0.0796 & -0.0001 & -0.0004 & 0.015 & 0.908 \\ 
  miss-\(Q\gamma\) & -0.0796 & 0.0753 & 0.0223 & 0.012 & 0.54 \\ 
  miss-\(Q\pi\) & -0.0796 & 0.0232 & -0.0004 & 0.0133 & 0.888 \\ 
  miss-\(\gamma\pi\) & -0.0796 & 0.1042 & 0.0202 & 0.0124 & 0.58 \\ 
  \bottomrule \\
\end{tabular} \\
\centering
Estimation of the overall interventional  disparity indirect effect: \\
\begin{tabular}{llllll}
  \toprule
  scenario & truth & bias (initial) & bias (tmle) & SE & Cov (95\%) \\ 
  \hline\\[-0.3cm]
  all-correct & -0.0556 & -0.001 & -0.0012 & 0.0151 & 0.956 \\ 
  miss-\(Q\) & -0.0556 & 0.0052 & -0.0012 & 0.0138 & 0.952 \\ 
  miss-\(\gamma\) & -0.0556 & 0.0514 & -0.001 & 0.0127 & 0.964 \\ 
  miss-\(\pi\) & -0.0556 & -0.001 & -0.0012 & 0.0149 & 0.974 \\ 
  miss-\(Q\gamma\) & -0.0556 & 0.0513 & 0.0057 & 0.0099 & 0.888 \\ 
  miss-\(Q\pi\) & -0.0556 & 0.0052 & -0.0183 & 0.0133 & 0.714 \\ 
  miss-\(\gamma\pi\) & -0.0556 & 0.0514 & 0.0525 & 0.0125 & 0.01 \\ 
  \bottomrule
\end{tabular}
\caption{Results from the simulation study \ref{eq:sim:1} with
  \(n=1000\) and \(M=500\) repetitions.}\label{tab:sim:results:1}
\end{table}

Next, we change our setting by introducing unmeasured
\(AZ\) confounding:
\begin{align}
    \label{eq:sim:2}\tag{sim-2}
  \begin{split}
  W_1 &\sim \mathrm{Ber}(0.6), \\
  W_2 &\sim U(-1,1), \\
    A & \sim \mathrm{Ber}(\mathrm{expit}(0.5-1.8W_1 +0.5 W_2^2)), \\
     U &  \sim \mathrm{Ber}(\mathrm{expit}(-0.9+1.7A)) \\
  Z & \sim \mathrm{Ber}(\mathrm{expit}(0.6-1.8W_1+0.5W_2^2-1.25U)), \\
    Y & \sim \mathrm{Ber}(\mathrm{expit}(-0.2-1.3W_1+W_2^2 + 0.8A  -0.6Z -1.8Z(1-W_1))); 
        \end{split}
\end{align}
notably, \(U\) is strongly associated with \(A\), and there is only an
effect of \(U\), not \(A\), on \(Z\). Now the simulation from the
counterfactual distribution:
\begin{align*}
  Z^{\gammapoor} & \sim \mathrm{Ber}(\mathrm{expit}(0.6-1.8 W_1+0.5  W_2^2 )), 
\end{align*}
will only identify the right hand side of
\eqref{eq:target:parameter:2} in Lemma \ref{interpret:causal:2}. To
estimate the true value of the statistical target parameter, now
interpreted only interventional according to Lemma
\ref{interpret:causal:1} can be done as follows: 
\begin{align*}
  U_0 & \sim \mathrm{Ber}(\mathrm{expit}(-0.9+1.7\cdot 0)) \\
  Z^{\gammapoor} & \sim \mathrm{Ber}(\mathrm{expit}(0.6-1.8 W_1+0.5  W_2^2 + U_0 )) \\
  Y^{\gammapoor} & \sim \mathrm{Ber}(\mathrm{expit}(-0.2-1.3W_1 +W_2^2
                 +0.8A -0.6 Z^{\gammapoor}-1.8 Z^{\gammapoor}(1-W_1))) . 
\end{align*}
In this particular setting, we have that
\(\EE [ Y^{\gamma^{a=0}} - Y \mid A=1 ] = 0\), and we estimate that
\(\EE [ Y^{\gammapoor} - Y \mid A=1 ] \approx -0.04398\), i.e., there is
an effect of the policy shifting the distribution of \(Z\) only
through the dependence on the unmeasured \(U\). In the estimation
procedure \(U\) is unknown, and Table \ref{tab:sim:results:2} shows
results for estimation of
\(\EE [ Y^{\gammapoor} - Y \mid A=1 ] \approx -0.04398\).

\begin{table}[ht]
  \centering
    Estimation of our target parameter: \\
\begin{tabular}{llllll}
  \toprule
  scenario & truth & bias (initial) & bias (tmle) & SE & Cov (95\%) \\ 
  \hline\\[-0.3cm]
  all-correct & -0.044 & 0.0008 & 0.001 & 0.0184 & 0.952 \\ 
  miss-\(Q\) & -0.044 & 0.0128 & 0.0007 & 0.0125 & 0.862 \\ 
  miss-\(\gamma\) & -0.044 & 0.0946 & 0.0012 & 0.0201 & 0.968 \\ 
  miss-\(\pi\) & -0.044 & 0.0008 & 0.0013 & 0.0128 & 0.838 \\ 
  miss-\(Q\gamma\) & -0.044 & 0.061 & 0.0159 & 0.0127 & 0.724 \\ 
  miss-\(Q\pi\) & -0.044 & 0.0128 & 0.0006 & 0.0105 & 0.788 \\ 
  miss-\(\gamma\pi\) & -0.044 & 0.0946 & 0.0152 & 0.0128 & 0.688 \\ 
  \bottomrule
\end{tabular}
\caption{Results from the simulation study \ref{eq:sim:2} (with
  unmeasured \(A\)-\(Z\) confounding) with \(n=1000\) and \(M=500\)
  repetitions.}\label{tab:sim:results:2}
\end{table}

Finally, we change our simulation setting by introducing support
differences for \(W_2\) on the exposed and unexposed. Particularly, we
now simulate data such that
\begin{align}
    \label{eq:sim:3}\tag{sim-3}
  \begin{split}
  W_1 &\sim \mathrm{Ber}(0.6), \\
  W_2 &\sim U(-1,1), \\
  A & \sim \mathrm{Ber}(\mathrm{expit}(0.5-1.8W_1 +0.5 W_2^2 - 4\1\lbrace W_2 > 0.5\rbrace)), \\
  Z & \sim \mathrm{Ber}(\mathrm{expit}(0.6-1.8W_1+0.5W_2^2-0.9A)), \\
    Y & \sim \mathrm{Ber}(\mathrm{expit}(-0.2-1.3W_1+W_2^2 + 0.8A  -0.6Z -1.8Z(1-W_1))).
        \end{split}
\end{align}
This now means that the support of \(W\) on the unexposed is larger
than on the exposed. This is not a problem for the estimation of our
parameter, see Assumption (A4), but it is a problem for the
interventional disparity indirect effect considered by
\citep{micali2018maternal}. Results can be found in Table
\ref{tab:sim:results:3:1} for estimation of the intervention specific
parameter \(\Psi_0(P)\) and \(\Psi_1(P)\), specifically, and in Table
\ref{tab:sim:results:3:2} for estimation of the disparity
effects. Note that we have here increased the sample size as well as
the number of simulation repetitions to ensure that the results found
are not simply due to monte carlo variation.

\begin{table}[ht]
  \centering
  Estimation of the risk under the interventional \(\gamma(z\mid 0,w)\): \\
\begin{tabular}{llllll}
  \toprule
on the exposed & 0.3553 & 0.0010 & 0.0008 & 0.0180 & 0.9460 \\ 
on the full population & 0.3692 & -0.0009 & -0.0012 & 0.0218 & 0.9300 \\
  \bottomrule \\
\end{tabular} \\
\centering
Estimation of the risk under the observed \(\gamma(z\mid 1,w)\): \\
\begin{tabular}{llllll}
  \toprule
  on the exposed & 0.4354 & 0.0005 & 0.0005 & 0.0177 & 0.9440 \\ 
  on the full population & 0.4240 & -0.0001 & 0.0001 & 0.0208 & 0.9230 \\ 
  \bottomrule
\end{tabular}
\caption{Results from the simulation study \ref{eq:sim:3} (with
  positivity violations) with \(n=2500\) and \(M=1000\)
  repetitions.}\label{tab:sim:results:3:1}
\end{table}

\begin{table}[ht]
  \centering
  Estimation of our target parameter: \\
\begin{tabular}{llllll}
  \toprule
  scenario & truth & bias (initial) & bias (tmle) & SE & Cov (95\%) \\ 
  \hline\\[-0.3cm]
all-correct & -0.0801 & 0.0005 & 0.0003 & 0.0125 & 0.951 \\ 
  miss-\(Q\) & -0.0801 & 0.0264 & 0.0003 & 0.0095 & 0.879 \\ 
  miss-\(\gamma\) & -0.0801 & 0.0988 & 0.0038 & 0.0117 & 0.932 \\ 
  miss-\(\pi\) & -0.0801 & 0.0005 & 0.0002 & 0.0108 & 0.896 \\ 
  miss-\(Q\gamma\) & -0.0801 & 0.0718 & 0.0283 & 0.007 & 0.071 \\ 
  miss-\(Q\pi\) & -0.0801 & 0.0264 & 0.0002 & 0.0093 & 0.869 \\ 
  miss-\(\gamma\pi\) & -0.0801 & 0.0988 & 0.0259 & 0.009 & 0.247 \\ 
  \bottomrule \\
\end{tabular} \\
\centering
Estimation of the overall interventional  disparity indirect effect: \\
\begin{tabular}{llllll}
  \toprule
  scenario & truth & bias (initial) & bias (tmle) & SE & Cov (95\%) \\ 
  \hline\\[-0.3cm]
all-correct & -0.0548 & -0.0009 & -0.0013 & 0.0106 & 0.93 \\ 
  miss-\(Q\)  & -0.0548 & 0.0071 & -0.0008 & 0.0104 & 0.915 \\ 
  miss-\(\gamma\) & -0.0548 & 0.0464 & 0.0028 & 0.0076 & 0.923 \\ 
  miss-\(\pi\) & -0.0548 & -0.0009 & -0.001 & 0.0109 & 0.961 \\ 
  miss-\(Q\gamma\) & -0.0548 & 0.0464 & 0.0109 & 0.0064 & 0.555 \\ 
  miss-\(Q\pi\) & -0.0548 & 0.0071 & -0.0187 & 0.0093 & 0.485 \\ 
  miss-\(\gamma\pi\) & -0.0548 & 0.0464 & 0.0395 & 0.0091 & 0.005 \\ 
  \bottomrule
\end{tabular}
\caption{Results from the simulation study \ref{eq:sim:3} with
  \(n=2500\) and \(M=1000\) repetitions.}\label{tab:sim:results:3:2}
\end{table}

\section{Extension to event history settings}
\label{sec:event:history:setting}

We consider the extension to a right-censored time-to-event outcome 
\((\tilde{T},\tilde{\Delta}) \in \R_+ \times \lbrace 0,1,\ldots,
J\rbrace\), where \(\tilde{T} = \min (C,T)\) is the minimum of a
(latent) censoring time \(C\in \R_+\) and a (latent) event time
\(T\in\R_+\) and \(\tilde{\Delta} = \1 \lbrace T\le C\rbrace \Delta\)
indicates right-censoring (\(\tilde{\Delta}=0\)) or type of event
(\(\Delta \ge 1\)). As before, \(W\in\R^d\) denotes a vector of
covariates, \(A\in\lbrace 0,1\rbrace\) is the binary exposure
variable, and \(Z\in \lbrace 0,1\rbrace\) is the binary intermediate
variable. For \(j=1,2\), we let \(\lambda_{0,j}\) denote the cause
\(j\) specific hazard defined as
\begin{align*}
  \lambda_{0,j}(t \, \vert \, z, a, w)
  & =
    \underset{h \rightarrow 0}{\lim} \,\, h^{-1} P(T \le t+h,  \Delta=j \mid {T} \ge t,
  Z=z,  A=a, W=w),  
\end{align*}
and \(\Lambda_{0,j}(t \mid z,a,w)\) the corresponding cumulative
hazard. Likewise, we let \(\lambda_0^c(t \mid z,a,w)\) denote the
conditional hazard for censoring and \(\Lambda_0^c(t \mid z,a,w)\) the
corresponding cumulative hazard. The survival function and the
censoring survival function are denoted
\( {S}_0 (t \mid z, a, w) = \exp (- \int_0^t \sum_{j=1}^J
{\lambda}_{0,j}(s \mid z, a, w) ds)\) and
\( {S}^c_0 (t \mid z, a, w) = \exp (- \int_0^t {\lambda}_0^c(s \mid z,
a, w) ds)\), respectively. Lastly, we denote by
\( {F}_{0,j}(t \mid a,w) = \int_0^t {S}_{0}(s- \mid a, w)
{\lambda}_{0,j} (s \mid a, w) ds\) the absolute risk function for
events of type \(j\) \citep{gray1988class}. As in Section
\ref{sec:setting:notation}, \(\pi (a\mid w)\) denotes the conditional
distribution of exposure \(A\) given covariates \(W\), and
\(\gamma(z \mid a, w)\) denotes the conditional distribution of the
intermediate variable \(Z\) given exposure \(A\) and covariates
\(W\). We assume throughout that
\begin{align*}
(T,\Delta) \independent C \mid (Z,A,W).
\end{align*}

\subsection{Target parameter and different interpretations}

Our target parameter is now defined as: 
\begin{align}
  \begin{split}
    \Psi_{a^*}(P)
    &= \EE \bigg[ \sum_{z=0,1} F_1 (\tau  \mid  Z=z, A, W)   \gamma(z \mid a^*, W)  \bigg\vert A=1\bigg] \\
    &=   \int_{\mathcal{W}} \sum_{z=0,1}  F_1 (\tau  \mid  Z=z, A, W) 
      \gamma( z \mid a^*, w) 
      \frac{\pi( 1 \mid w)}{\bar{\pi}(1)}  d\mu (w) \\
    &=   \int_{\mathcal{W}} \sum_{z=0,1} \bigg( \int_0^\tau \lambda_1 (t \mid z, 1, w) S(t-\mid z,1,w)dt\bigg) 
      \gamma( z \mid a^*, w) 
      \frac{\pi( 1 \mid w)}{\bar{\pi}(1)}  d\mu (w),
      \end{split} \label{target:parameter:cr}
\end{align}
corresponding to the intervention-specific risk evaluated at time
\(\tau >0\).

The following structural assumptions correspond to these needed for an
interventional interpretation from in Lemma \ref{interpret:causal:1}
of our target parameter in the uncensored setting:
\begin{enumerate}
\item[(SA1)] \((T^z, \Delta^z) \independent Z \mid (A, W)\), for \(z=0,1\);
\item[(SA2)] \(T^{z} = T\) and \(\Delta^z=\Delta\) if \(Z=z\), for
  \(z=0,1\);
\item[(SA3)] \(P(Z=z \mid A=a, W)>\eta>0\) for \(z=0,1\) and \(a=0,1\); 
\item[(SA4)] \(P(A=a \mid W)>\eta>0\) for \(a=0,1\); and
\item[(SA5)] \(S^c( t\mid z,a,W) > \eta>0\) for \(z=0,1\), \(a=0,1\)
  and \(t<\tau\).
\end{enumerate}
Define 
\((T^{{\gammapoor}}, \Delta^{{\gammapoor}}) \) as the counterfactual
outcome we would observe had we intervened and changed the conditional
distribution of \(Z\) to be \(\gammapoor\) rather than \(\gamma\);
similarly define \((T^{{\gamma_{a=1}}}, \Delta^{{\gamma_{a=1}}}) \).  The
statistical parameter defined by \eqref{target:parameter:cr}
identifies the causal parameter
\begin{align}
  \Psi(P) =  P( T^{{\gammapoor}} \le \tau, \Delta^{{\gammapoor}}=1 \mid A=1)
  - P( T^{{\gamma_{a=1}}} \le \tau, \Delta^{{\gamma_{a=1}}}=1\rbrace \mid A=1 ), 
  \label{eq:target:parameter:cr:1}
\end{align}
under Assumptions (SA1)--(SA5). To obtain the stronger interpretation
as in Lemma \ref{interpret:causal:2}, i.e., that
\begin{align}
  \Psi(P) =  P ( T^{{\gamma^{a=0}}} \le \tau, \Delta^{{\gamma^{a=0}}}=1 \mid A=1) 
  - P( T^{{\gamma^{a=1}}} \le \tau, \Delta^{{\gamma^{a=1}}}=1 \mid A=1 ), 
  \label{eq:target:parameter:cr:2}
\end{align}
relies on further Assumptions (A1*) and (A2*) as stated just before
Lemma \ref{interpret:causal:2}.

\subsection{Statistical estimation problem}

Below we present the efficient influence function and the double
robustness properties for estimation of the target parameter in the
event history setting. The biggest differences lie in the fact that
the censoring mechanism acts continuously in term, so that first term
of the efficient influence curve becomes a (martingale) integral
across time. Moreover, the double robustness properties change to
involve the censoring mechanism across time as well.

\begin{thm}[Efficient influence function for the event history setting]
  The efficient influence function for
  \(\Psi_{a^*} \, : \, \mathcal{M}\rightarrow \R\) can be represented
  as follows:
\begin{align*}
  \phi_{a^*} (P) (O)
  &  = \frac{\gamma( Z \mid a^*, W) }{\gamma( Z \mid 1, W)}
    \frac{\1\lbrace A=1\rbrace}{\bar{\pi}(1)} 
  \bigg(  \int_0^{\tau}
    \big( S^c(t- \mid Z,A,W)\big)^{-1} \\
  &\qquad\qquad \times \bigg( 1-
    \frac{F_1 (\tau \mid Z,A,W) - F_1(t \mid Z,A,W)}{S(t \mid Z,A,W)}\bigg)\big(N_1 (dt) - \1\lbrace
    \tilde{T}\ge t \rbrace \lambda_1 (t \mid Z,A,W) dt \big) \\
 &\qquad\qquad -\sum_{l\neq 1} \int_0^{\tau}
    \big( S^c(t- \mid Z,A,W)\big)^{-1} \bigg( 
   \frac{F_1 (\tau \mid Z,A,W)- F_1(t \mid Z,A,W)}{S(t \mid Z,A,W)}\bigg) \\[-0.3cm]
    &\qquad\qquad\qquad\qquad\qquad\qquad\qquad\qquad\qquad\qquad\qquad \times \big(N_l (dt) - \1\lbrace
  \tilde{T}\ge t \rbrace \lambda_l (t \mid Z,A,W) dt \big) \bigg) \\[-0.1cm]
                    &\qquad\quad  + \frac{\1\lbrace A=a^*\rbrace}{\pi(A \mid W) } \frac{\pi( 1 \mid W)}{\bar{\pi}(1)}
                      \bigg( F_1(\tau \mid Z, 1, W) - \sum_{z=0,1} F_1 (\tau \mid z, 1, W) \gamma(z \mid A , W) \bigg) \\
                    &\qquad\quad + \frac{\1\lbrace A=1\rbrace }{\bar{\pi}(1)} \bigg(\sum_{z=0,1} F_1 (\tau \mid z, 1, W) \gamma(z \mid a^* , W)
                      - \Psi_{a^*} (P) \bigg)  .
\end{align*}
We also write
\(\phi_{a^*} (P) = \tilde{\phi}_{a^*} (\pi, \bar{\pi}, \gamma, \lambda_1,\ldots, \lambda_J, S^c)\).
\label{thm:eff:ic:survival}
\end{thm}

\begin{proof}
  See Appendix \ref{app:eff:ic:survival}.
\end{proof}

\begin{thm}[Double robustness in the event history setting]
  Say we estimate \(\pi\) by \(\hat{\pi}_n\), \(S^c\) by
  \(\hat{S}^c_n\), \(\gamma\) by \(\hat{\gamma}_n\) and \(\lambda_j\)
  by \(\hat{\lambda}_{n,j}\). Let \({\pi}'_0\), \({\gamma}'_0\), \({{S}_0^c}'\) and
  \({\lambda}'_{0,j}\) denote the limits of these estimators. Let
  \(\hat{\bar{\pi}}_n = \frac{1}{n} \sum_{i=1}^n \1\lbrace
  A_i=1\rbrace\). Assume first that
  \(\mathbb{P}_n \tilde{\phi}_{0} (\hat{\pi}_n, \hat{\bar{\pi}}_n,
  \hat{\gamma}_n,\hat{\lambda}_{n,1},\ldots, \hat{\lambda}_{n,J},
  \hat{S}^c_n)=o_P(n^{-1/2})\). Then
  \(\tilde{\Psi}_0(\hat{\pi}_n, \hat{\bar{\pi}}_n,
  \hat{\gamma}_n,\hat{\lambda}_{n,1},\ldots, \hat{\lambda}_{n,J})\)
  provides a consistent estimator for \(\Psi_{0} (P_0)\) if:
  \begin{enumerate}
  \item[a.] \({{S}_0^c}'= S^c_0 \) and \({\gamma}'_0 = \gamma_0\); or
  \item[b.] \({\pi}'_0= \pi_0 \) and
    \({\lambda}'_{0,1} = {\lambda}_{0,1} , \ldots, {\lambda}'_{0,J} =
    {\lambda}_{0,J} \).
  \end{enumerate}
  Similarly,
  \(\tilde{\Psi}_1(\hat{\lambda}_{n,1},\ldots, \hat{\lambda}_{n,J})\)
  provides a consistent estimator for \(\Psi_{1} (P_0)\) if:
  \begin{enumerate}
  \item[c.] \({{S}_0^c}'= S^c_0 \); or
  \item[d.]
    \({\lambda}'_{0,1} = {\lambda}_{0,1} , \ldots, {\lambda}'_{0,J} =
    {\lambda}_{0,J} \).
\end{enumerate}
\label{thm:double:robustness:survival}
\end{thm}

\begin{proof}
  Define
  \(R_{a^*}(P, P_0) = \Psi_{a^*}( P) - \Psi_{a^*}(P_0) + P_0
  \phi_{a^*}(P)\).
  We show in Appendix \ref{app:remainder:survival} that
  \begin{align*}
 &   R_{1}(P,P_0)    =
                     \EE_{P_0}
    \bigg[\frac{\pi_0(1 \mid W)}{\bar{\pi}(1)} \sum_{z=0,1} 
    \gamma_0 ( z \mid 1, W)
    \bigg(  \int_0^{\tau} \bigg( \frac{ S_0^c(t- \mid z,1,W) -   S^c(t- \mid z,1,W)}{ S^c(t- \mid z,1,W) }  \bigg)  \notag \\
  &\qquad\qquad \times S_0(t- \mid z,1,W) \bigg( 1-
    \frac{F_1 (\tau \mid z,1,W) - F_1(t \mid z,1,W)}{S(t \mid z,1,W)}\bigg)\big( \Lambda_{0,1} (dt \mid z,1,W)  -  \Lambda_1 (dt \mid z,1,W)  \big) \notag \\
  &\qquad\qquad\qquad\qquad -\sum_{l\neq 1} \int_0^{\tau}
    \frac{S_0^c(t- \mid z,1,W)}{ S^c(t- \mid z,1,W) } S_0(t- \mid z,1,W)\bigg( 
    \frac{F_1 (\tau \mid z,1,W)- F_1(t \mid z,1,W)}{S(t \mid z,1,W)}\bigg) \notag \\[-0.3cm]
  &\qquad\qquad\qquad\qquad\qquad\qquad\qquad\qquad\qquad\qquad \times \big(\Lambda_{0,l} (dt \mid z,1,W) -
    \Lambda_l (dt \mid z,1,W)  \big) \bigg)   \bigg] \\
    &\qquad +  
                         \bigg( 1-     \frac{\bar{\pi}_0( 1)}{\bar{\pi}(1)} \bigg) \big( \Psi_{1} (P)
      - \Psi_{1} (P_0) \big)
      \intertext{and,}
&   R_{0}(P,P_0)         =
    \EE_{P_0}
    \bigg[\frac{\pi_0(1 \mid W)}{\bar{\pi}(1)} \sum_{z=0,1} 
    \gamma ( z \mid 0, W)
    \bigg(  \int_0^{\tau} \bigg( \frac{\gamma_0( z \mid 1, W) S_0^c(t- \mid z,1,W) - \gamma( z \mid 1, W) S^c(t- \mid z,1,W) }{\gamma( z \mid 1, W) S^c(t- \mid z,1,W)}
     \bigg) \notag \\
  &\qquad\qquad \times  S_0(t- \mid z,1,W)\bigg( 1-
    \frac{F_1 (\tau \mid z,1,W) - F_1(t \mid z,1,W)}{S(t \mid z,1,W)}\bigg)\big( \Lambda_{0,1} (dt \mid z,1,W)  -  \Lambda_1 (dt \mid z,1,W)  \big) \notag \\
  &\qquad\qquad\qquad\qquad -\sum_{l\neq 1} \int_0^{\tau}
    \frac{S_0^c(t- \mid z,1,W)}{ S^c(t- \mid z,1,W) } S_0(t- \mid z,1,W)\bigg( 
    \frac{F_1 (\tau \mid z,1,W)- F_1(t \mid z,1,W)}{S(t \mid z,1,W)}\bigg) \notag \\[-0.3cm]
  &\qquad\qquad\qquad\qquad\qquad\qquad\qquad\qquad\qquad\qquad \times \big(\Lambda_{0,l} (dt \mid z,1,W) -
    \Lambda_l (dt \mid z,1,W)  \big) \bigg)   \bigg] \\
  & \qquad + \EE_{P_0}  \bigg[  \bigg(  \frac{\pi_0(1 \mid W)  - \pi( 1 \mid W)
                }{\bar{\pi} (1) ( 1- \pi( 1 \mid W))} \bigg) 
  \sum_{z=0,1} F_{1}(\tau \mid z, 1, W) \big( \gamma(z \mid 0 , W)-  \gamma_0(z \mid 0 , W) \big)
    \bigg] \\
& \qquad +
  \EE_{P_0}  \bigg[ \frac{\pi_0(1 \mid W) }{\bar{\pi} (1)} 
  \sum_{z=0,1}  \big( F_{0,1}(\tau \mid z, 1, W)- F_{1}(\tau \mid z, 1, W) \big) \big( \gamma(z \mid 0 , W)-  \gamma_0(z \mid 0 , W) \big)
  \bigg]  \\
    &\qquad +  
                         \bigg( 1-     \frac{\bar{\pi}_0( 1)}{\bar{\pi}(1)} \bigg) \big( \Psi_{0} (P)
      - \Psi_{0} (P_0) \big). 
\end{align*}
from which a., b., c. and d.  follow.
\end{proof}

\subsection{Targeting algorithm}
\label{sec:survival:targeting}

In the event history setting we need to target both \(\Psi_0(P_0)\)
and \(\Psi_1(P_0)\), although the latter will be simpler. Indeed, note
that for \(a^*=1\) we have that
\begin{align*}
  \phi_{1} (P) (O)
  &  = 
    \frac{\1\lbrace A=1\rbrace}{\bar{\pi}(1)} 
    \bigg(  \int_0^{\tau}
    \big( S^c(t- \mid Z,A,W)\big)^{-1} \\
  &\qquad\qquad \times \bigg( 1-
    \frac{F_1 (\tau \mid Z,A,W) - F_1(t \mid Z,A,W)}{S(t \mid Z,A,W)}\bigg)\big(N_1 (dt) - \1\lbrace
    \tilde{T}\ge t \rbrace \lambda_1 (t \mid Z,A,W) dt \big) \\
  &\qquad\qquad -\sum_{l\neq 1} \int_0^{\tau}
    \big( S^c(t- \mid Z,A,W)\big)^{-1} \bigg( 
    \frac{F_1 (\tau \mid Z,A,W)- F_1(t \mid Z,A,W)}{S(t \mid Z,A,W)}\bigg) \\[-0.3cm]
  &\qquad\qquad\qquad\qquad\qquad\qquad\qquad\qquad\qquad\qquad\qquad \times \big(N_l (dt) - \1\lbrace
    \tilde{T}\ge t \rbrace \lambda_l (t \mid Z,A,W) dt \big) \bigg) \\[-0.1cm]
  &\qquad\quad  + \frac{\1\lbrace A=1\rbrace}{\bar{\pi}(1)}
    \big( F_1(\tau \mid Z, 1, W)
    - \Psi_{1} (P) \big)  , 
\end{align*}
so that targeting for \(\Psi_1(P_0)\) will require updating
\(\lambda_1,\ldots, \lambda_J\) such as to take care of the first term
(first four lines) above. Define
\begin{align}
  &H^{a^*}_{j,t}  (\lambda_1, \ldots, \lambda_J, \lambda^c, \pi) (Z,A,W) \notag\\
    \begin{split}
    &\qquad  = \begin{cases}
    \frac{\gamma( Z \mid a^*, W) }{\gamma( Z \mid 1, W)}
    \frac{\1\lbrace A=1\rbrace}{\bar{\pi}(1)} \frac{1}{
      {S}^c( t- \, \vert \,A, L)
    } \Big( 1- \frac{ {F}_1 ( \tmax  \mid Z, A, W) - {F}_1 ( t \mid Z, A, W)}{
      {S}(t \, \vert \, A,L)} \Big) , & \text{for } j=1, \\
    \frac{\gamma( Z \mid a^*, W) }{\gamma( Z \mid 1, W)}
    \frac{\1\lbrace A=1\rbrace}{\bar{\pi}(1)} \frac{1}{
      {S}^c( t- \, \vert \,Z, A, W)
    } \Big(
    -\frac{ {F}_1 ( \tmax  \mid Z, A, W) - {F}_1 ( t \mid Z, A, W)}{
      {S}(t \, \vert \, A,L) }
    \Big) , & \text{for } j \neq 1,
                                                         \end{cases}
      \end{split}\label{eq:surv:eff:ic:continuous:1:covars}
\end{align}
and note that the first term of the efficient influence function
presented in Theorem \ref{thm:eff:ic:survival} can then be written
\begin{align}
  \int_0^\tau   H^{a^*}_{j,t}  (\lambda_1, \ldots, \lambda_J, \lambda^c, \pi) (Z,A,W)\big( N_j(dt) -
  \1\lbrace \tilde{T}\ge t\rbrace \lambda_j(dt \mid Z,A,W) dt \big) .
  \label{eq:surv:eic:term}
\end{align}
Targeting may now proceed in an iterative manner as proposed by
\cite{rytgaard2021estimation}, updating each cause-specific hazard
\(\lambda_j\) along the multiplicative fluctuation model through
\(\lambda_j\) at \(\eps_j = 0\) defined as
\begin{align}
  \lambda_{j,\eps_j} (t \mid z, a, w)
  & = \lambda_j(t \mid z, a, w)
    \exp ( \eps_j H^{a^*}_{j,t} (\lambda_j,\ldots, \lambda_J, \lambda^c, \pi) (z,a,w) ), \qquad
    \eps_j\in\R, 
  \label{eq:lambda:fluc}
\end{align}
with fluctuation parameter \( \eps_j \in \R\).  With
\((O, \lambda) \mapsto \ell_{\mathrm{loglik}}(\lambda_j) (O)\)
denoting the log-likelihood function for \(\lambda_j\),
\begin{align}
  \ell_{\mathrm{loglik}}(\lambda_j) (O) = \int_0^{\tmax} \log \lambda_j (t \mid Z, A, W)\, N_j(dt) -
  \int_0^{\tmax} \1\lbrace \tilde{T} \ge t\rbrace \lambda_j (t \mid Z, A, W) dt , 
  \label{eq:def:loglik}
\end{align}
we have the property that
\begin{align}
         \begin{split}
           & \frac{d}{d\eps} \bigg\vert_{\eps = 0}
             \ell_{\mathrm{loglik}}(\lambda_{j,\eps}) (O) =
             \int_0^{\tmax}  H^{a^*}_{j,t} (\lambda_1, \ldots, \lambda_J, \lambda^c, \pi)
             (A,L) \, N_j(dt) \\[-0.5em]
           & \qquad\qquad\qquad\qquad\qquad
             - \, \int_0^{\tmax} \1\lbrace \tilde{T} \ge t\rbrace  H^{a^*}_{j,t} (\lambda_1, \ldots, \lambda_J, \lambda^c, \pi)
             (A,L) \, {\lambda}_{j,\eps} (t \mid Z, A, W ) dt,
  \end{split}  \label{eq:tmle:desired:property}
\end{align}
recognizing the right hand side as the relevant \(j\)th term
\eqref{eq:surv:eic:term} of the efficient influence curve. The
following targeting steps (across \(j\)) are now substituted for the
step updating \(\hat{Q}^k_n\) in Section \ref{sec:targeting:steps}:

\begin{description}
\item[Updating \(\hat{\lambda}^k_{j,n}\):] Estimate \(\eps_j\) by the
  maximum likelihood estimator
  \(\hat{\eps}_{j,n}^{k} = \text{argmax}_\eps \mathbb{P}_n
  \ell_{\mathrm{loglik}}(\hat{\lambda}^k_{j,n,\eps}) \) for the
  fluctuation parameter along the submodel \eqref{eq:lambda:fluc}. 
\end{description}
  
\noindent To target \(\Psi_0(P_0)\), the updating steps for estimators
for \(\gamma\) and \(\pi\) are carried out as in Section
\ref{sec:targeting:steps}, except with the clever covariates now
defined by
\begin{align}
  H_2 (Q,\pi) ( W) & :=  \frac{ \1\lbrace A=0\rbrace }{\pi (0\mid W)} \frac{\pi(1 \mid W) }{\bar{\pi}(1)}
                   \big( F_1(\tau \mid 1,1,W) - F_1(\tau \mid 0, 1, W)  \big),\label{eq:clever:2:surv}\\
  H_3 (Q, \gamma) ( W) & :=   \frac{ \sum_{z=0,1} F_1(\tau \mid z, 1, W)  
    \gamma(z \mid 0, W) - \Psi_0(P) }{\bar{\pi}(1)} .\label{eq:clever:3:surv} 
\end{align}
To target \(\Psi_1(P_0)\), only \(\hat{\lambda}^k_{1,n}\), \ldots,
\(\hat{\lambda}^k_{J,n}\) are updated.  The iterations are continued
until the final set of updated estimators, i.e.,
\(\hat{P}^{*,0}_n = (\hat{\lambda}^{*,0}_{1,n}, \ldots,
\hat{\lambda}^{*,0}_{J,n}, \hat{Q}^{*,0}_n, \hat{\gamma}^{*,0}_n,
\hat{\pi}^{*,0}_n,\hat{\bar{\pi}}_n)\) and
\(\hat{P}^{*,1}_n = (\hat{\lambda}^{*,1}_{1,n}, \ldots,
\hat{\lambda}^{*,1}_{J,n}, \hat{Q}^{*,1}_n, \hat{\gamma}^{*,1}_n,
\hat{\pi}^{*,1}_n,\hat{\bar{\pi}}_n)\), respectively, solve the
efficient influence curve equation sufficiently well.

\section{Discussion}
\label{sec:discussion}

With a focus on applications where policy interventions on
intermediate variables are of interest to reduce health disparities,
we have defined, studied and discussed a specific target estimand
arising under a stochastic intervention modifying the distribution of
the intermediate variable. We have further presented a theoretical
analysis and discussion of nonparametric estimation of this
effect. Finally, we have extended the theoretical analysis and
statistical estimation procedures to include right-censored
time-to-event outcomes. We hope that these contributions can help
inform policy decisions aimed at reducing health disparities among
different population groups defined by variables that cannot
meaningfully be subjected to intervention.

Our results show important double robustness properties of the
parameter that are not immediately trivial; indeed, our parameter
depends on all nuisance parameters, but in the end still allows for
misspecification as long as it is done within the constraints
presented in Theorem \ref{thm:double:robustness}. In our simulations
we verified these properties, and also found (which would need to be
theoretically verified) seemingly different robustness properties for
the corresponding interventional disparity indirect measure among the
full population considered by \cite{micali2018maternal}. Our
simulation study further highlights the differences in positivy
requirements for the intervention disparity indirect effect among the
full population and among the exposed population, and the potential
for more precise estimation when the support of \(W\) on the exposed
is contained in the unexposed. Finally, our simulations study
demonstrates the difference between the interventional interpretations
presented in Lemmas \ref{interpret:causal:1} and
\ref{interpret:causal:2}.

For the event history setting, we considered an iterative targeting
procedure along (local) least favorable submodel as proposed in
\cite{rytgaard2021estimation}.  Alternatively, the targeted could be
carried out along universal least favorable submodels
\citep{rytgaard2021one,rytgaard2022targetedlida} to target
multivariate parameter; this is often of critical interest in survival
and competing risks analysis when analyzing multivariate
parameters.

\newpage\appendix

\counterwithin*{equation}{section}
\renewcommand{\theequation}{\thesection.\arabic{equation}}

\section{Proofs of Lemma \ref{interpret:causal:1} and Lemma
  \ref{interpret:causal:2}}
\label{app:causal:proofs}

\subsection{Proof of Lemma \ref{interpret:causal:1}}
\label{app:lemma:1:proof}

The proof is straightforward. Starting from the right hand side of
\eqref{eq:target:parameter:1}, we have by Assumptions (A1)--(A4) that:
\begin{align*}
  &\EE [ Y^{{\gammapoor}}  - Y^{{\gamma_{a=1}}}  \mid A=1 ] \\
  & \quad = \EE \bigg[ \sum_{z=0,1} Y^z \big( \gammapoor(z \mid W) -  \gamma_{a=1} (z \mid  W)\big) \,\bigg\vert\, A=1\bigg] \\
  & \quad = \EE \bigg[\EE \bigg[ \sum_{z=0,1} Y^{z}  \big( \gamma(z \mid 0, W) -  \gamma (z \mid 1, W)\big) 
    \,\bigg\vert\, A=1, W\bigg]\,\bigg\vert\, A=1\bigg] \\
  & \quad = \EE \bigg[\sum_{z=0,1}  \EE \big[ Y^{z}  \,\big\vert\, A=1, W\big]
    \big( \gamma(z \mid 0, W) -  \gamma (z \mid 1, W)\big) 
    \,\bigg\vert\, A=1\bigg] \\
  & \quad \overset{(A1)}{=} \EE \bigg[\sum_{z=0,1}  \EE \big[ Y^{z}  \,\big\vert\, A=1, Z=z, W\big]
    \big( \gamma(z \mid 0, W) -  \gamma (z \mid 1, W)\big) 
    \,\bigg\vert\, A=1\bigg] \\
  & \quad \overset{(A2)}{=}\EE \bigg[\sum_{z=0,1}  \EE \big[ Y  \,\big\vert\, A=1,Z=z, W\big]
    \big( \gamma(z \mid 0, W) -  \gamma (z \mid 1, W)\big) 
    \,\bigg\vert\, A=1\bigg] , 
\end{align*}
which is seen to equal \(\Psi(P)\). Note that the positivity
assumptions (A3)--(A4) are needed to make sure all conditional
expectations are well-defined. This establishes
\eqref{eq:target:parameter:1}.

\subsection{Proof of Lemma \ref{interpret:causal:2}}
\label{app:lemma:2:proof}

  The proof is again straightforward. We start from the right hand
  side of \eqref{eq:target:parameter:2}:
\begin{align*}
  & \EE [ Y^{\gamma^{a=0}}  - Y^{\gamma^{a=1}}  \mid A=1 ] \\
  & \quad = \EE \bigg[ \sum_{z=0,1} Y^z \big(  P(Z^0 =z \mid W) -  P(Z^1 =z \mid W)\big)
    \,\bigg\vert\, A=1\bigg] \\
  & \quad = \EE \bigg[\EE \bigg[ \sum_{z=0,1} Y^{z}  \big(  P(Z^0 =z \mid W) -  P(Z^1 =z \mid W)\big) 
    \,\bigg\vert\, A=0, W\bigg]\,\bigg\vert\, A=1\bigg] \\
  & \quad = \EE \bigg[\sum_{z=0,1}  \EE \big[ Y^{z}  \,\big\vert\, A=1, W\big]
    \big(  P(Z^0 =z \mid W) -  P(Z^1 =z \mid W)\big) 
    \,\bigg\vert\, A=1\bigg] \\
  & \quad \overset{(A1)}{=} \EE \bigg[\sum_{z=0,1}  \EE \big[ Y^{z}  \,\big\vert\, A=1, Z=z, W\big]
    \big(  P(Z^0 =z \mid W) -  P(Z^1 =z \mid W)\big) 
    \,\bigg\vert\, A=1\bigg] \\
  & \quad \overset{(A2)}{=}\EE \bigg[\sum_{z=0,1}  \EE \big[ Y  \,\big\vert\, A=1,Z=z, W\big]
    \big(  P(Z^0 =z \mid W) -  P(Z^1 =z \mid W)\big) 
    \,\bigg\vert\, A=1\bigg] \\
  & \quad \overset{(A1*)}{=}\EE \bigg[\sum_{z=0,1}  \EE \big[ Y  \,\big\vert\, A=1,Z=z, W\big]
    \big(  P(Z^0 =z \mid A=0, W) -  P(Z^1 =z \mid A=1, W)\big)
    \,\bigg\vert\, A=1\bigg] \\
  & \quad \overset{(A2*)}{=}\EE \bigg[\sum_{z=0,1}  \EE \big[ Y  \,\big\vert\, A=1,Z=z, W\big]
    \big(  P(Z =z \mid A=0, W) -  P(Z =z \mid A=1, W)\big)     
    \,\bigg\vert\, A=1\bigg] \\
  & \quad =\EE \bigg[\sum_{z=0,1}  \EE \big[ Y  \,\big\vert\, A=1,Z=z, W\big]
    \big(  \gamma (z \mid 0, W) -   \gamma (z \mid 1, W)\big)
    \,\bigg\vert\, A=1\bigg]  . 
\end{align*}
This establishes \eqref{eq:target:parameter:2}.

\section{Proof Theorem \ref{thm:inference:tmle}}
\label{app:decomposition:proof}

By definition of the second-order remainder, see the proof of Theorem
\ref{thm:double:robustness}, we have that
\begin{align*}
  &  \Psi_{a^*}(P) - \Psi_{a^*}(P_0) \\
  & \qquad = 
    -  P_0 \phi^*(P) + R_2(P,P_0) \\
  & \qquad = 
    (\mathbb{P}_n - P_0 ) ( \phi^*(P) - \phi^* (P_0) ) 
    + 
    ( \mathbb{P}_n - P_0 )  \phi^* (P_0)
    - 
    \mathbb{P}_n \phi^*(P)
    +
    R_2(P,P_0). 
\end{align*}
Specifically, substituting the estimator \(\hat{P}^*_n\) for \(P\) in
the above expression gives a decomposition for the estimator
\(\hat{\psi}^{\mathrm{tmle}}_{a^*,n} = \Psi_{a^*}(\hat{P}^*_n)
\),
\begin{align}
  \Psi_{a^*}(\hat{P}^*_n) - \Psi_{a^*}(P_0)
  & =
    \mathbb{P}_n \phi^* (P_0) \notag \\
  &\quad - P_0   \phi^* (P_0) -  \mathbb{P}_n \phi^*(\hat{P}^*_n)  \notag\\
  &\quad   + \,
    (\mathbb{P}_n - P_0 ) ( \phi^*(\hat{P}^*_n) - \phi^* (P_0) ) \label{eq:term1} \\
  &\quad   + \,
    R_{a^*}(\hat{P}^*_n,P_0). \label{eq:term2} 
\end{align}
By definition we have that \(P_0 \phi^* (P_0) = 0\).  The iterations
of the targeting procedure (Section \ref{sec:targeting:steps}) are
repeated until the efficient influence curve equation is solved up to
factor \(1/(\sqrt{n}\log n)\); this ensures that \(\hat{P}^*_n\) used
to construct the TMLE estimator solves the efficient influence curve
equation, i.e., \( \mathbb{P}_n \phi^*(\hat{P}^*_n) =
o_P(n^{-1/2})\). The remaining terms \eqref{eq:term1} and
\eqref{eq:term2} are \(o_P(n^{-1/2})\) by Assumptions (R1) and
(R2). This finishes the proof.

\section{Derivation of the efficient influence functions}
\label{app:eff:ic}

\subsection{Binary outcome setting}
\label{app:eff:ic:binary}

We derive the efficient influence functions for the two parameters
\(\Psi_0(P)\) and \(\Psi_1(P)\) separately. Fix \(a^*=0,1\) and
consider the parameter \(\Psi_{a^*} \, :\, \mathcal{M}\rightarrow\R\),
i.e.,
\begin{align*}
  \Psi_{a^*}(P)
  &= \EE \bigg[ \sum_{z=0,1} \EE[ Y \mid  Z=z, A, W]   \gamma(z \mid a^*, W)  \bigg\vert A=1\bigg] \\
  &=   \int_{\mathcal{W}} \sum_{z=0,1}  Q (z, 1, w) 
    \gamma( z \mid a^*, w) 
    \frac{\pi( 1 \mid w)}{\bar{\pi}(1)}  d\mu (w) \\
  &=   \int_{\mathcal{W}} \sum_{z=0,1}  \int_{\mathcal{Y}} y \, dP_{Y} (y  \mid z, 1, w) 
    \gamma( z \mid a^*, w) 
    \frac{\pi( 1 \mid w)}{\bar{\pi}(1)}  d\mu (w).
\end{align*}
We first make some general observations. The observed data density
factorizes as follows
\begin{align}
  p( o \mid A=1) = p_{Y} (y \mid z, a, w) \gamma (z \mid a, w) \frac{\pi (a \mid w)}{\bar{\pi}(1)} \mu (w), 
  \label{eq:factorization:density}
\end{align}
and we define the parametric submodel \(P_{\eps}\) with density
\begin{align*}
  p_{\eps} (o ) = (1+\eps \mathcal{S}(o)) p_{\eps} (o),
\end{align*}
and score \(\mathcal{S} (o) = \frac{d}{d\eps}\vert_{\eps =0 }\log p_{\eps} (o)\).
The efficient influence function
\(\phi_{a^*} (P) \, :\, \mathcal{O} \rightarrow \R\) of the target
parameter \(\Psi_{a^*} \, :\, \mathcal{M}\rightarrow\R\) relative to
the maximal tangent space for the nonparametric model is a function in
this tangent space for which
\begin{align*}
 \frac{d}{d\eps}\bigg\vert_{\eps = 0} \Psi_{a^*}(P_{\eps}) = P \phi_{a^*} (P) \mathcal{S}, 
\end{align*}
for every \(\mathcal{S}\) in the tangent space and parametric model \(P_{\eps}\)
with score \(\mathcal{S}\). Since the observed data density factorizes according
to \eqref{eq:factorization:density}, the score \(\mathcal{S}\) decomposes into
an orthogonal sum of factor-specific
scores,
\begin{align*}
  S =
  \underbrace{\frac{d}{d\eps}\bigg\vert_{\eps =0 } \log p_{Y,\eps}}_{=\mathcal{S}_Y}
  +  \underbrace{ \frac{d}{d\eps}\bigg\vert_{\eps =0 } \log \gamma_{\eps}}_{=\mathcal{S}_{\gamma}}
  +  \underbrace{ \frac{d}{d\eps}\bigg\vert_{\eps =0 } \log \pi_{\eps}}_{=\mathcal{S}_{\pi}} 
  -  \underbrace{ \frac{d}{d\eps}\bigg\vert_{\eps =0 } \log P_{A,\eps}}_{=\mathcal{S}_{\bar{\pi}}}
  +  \underbrace{ \frac{d}{d\eps}\bigg\vert_{\eps =0 } \log \mu_{\eps}}_{=\mathcal{S}_{\mu}}  
\end{align*}
where
\begin{align}
  \mathcal{S}_{j} (O) = \Pi (\mathcal{S} (O) \mid T_{j} (P) )
  = \EE [ \mathcal{S}(O) \mid X_j, \mathrm{Pa}(X_j)] -
  \EE [ \mathcal{S}(O) \mid  \mathrm{Pa}(X_j)],
    \label{eq:Sj:projection}
\end{align}
with \( \Pi (\mathcal{S} (O) \mid T_{j} (P) )\) denoting the projection of
\(\mathcal{S}(O)\) onto the tangent space \(T_{j} (P)\) corresponding to the
\(j\)th factor, and \(\mathrm{Pa}(X_j)\) denoting the observed
variables included in the conditioning set for the \(j\)th
factor. Define
\begin{align}
  p_{X_j, \eps} (X_j \mid \mathrm{Pa}(X_j)) = (1+\eps \mathcal{S}_j(O)) p_{X_j} (X_j \mid \mathrm{Pa}(X_j)),
  \label{eq:direction}
\end{align}
so that
\(\mathcal{S}_j (O) = \frac{d}{d\eps}\vert_{\eps =0 } \log p_{X_j,\eps} (X_j \mid
\mathrm{Pa}(X_j))\). Combining \eqref{eq:Sj:projection} and
\eqref{eq:direction}, it is clear that
\begin{align*}
  \frac{d}{d\eps}\bigg\vert_{\eps = 0}  p_{X_j, \eps} (X_j \mid \mathrm{Pa}(X_j))
  &=
    p_{X_j, \eps} (X_j \mid \mathrm{Pa}(X_j))   \frac{d}{d\eps}\bigg\vert_{\eps = 0} \log p_{X_j, \eps} (X_j \mid \mathrm{Pa}(X_j))\\
  &  =  p_{X_j, \eps} (X_j \mid \mathrm{Pa}(X_j)) \mathcal{S}_j(O) \\[0.2cm]
  &  =  p_{X_j, \eps} (X_j \mid \mathrm{Pa}(X_j)) \big( \EE [ \mathcal{S}(O) \mid X_j, \mathrm{Pa}(X_j)] -
    \EE [ \mathcal{S}(O) \mid  \mathrm{Pa}(X_j)]\big) , 
\end{align*}
which we will use in all the following calculations. Specifically,
note that
\begin{align*}
  \frac{d}{d\eps}\bigg\vert_{\eps = 0} \frac{1}{ P_{A, \eps} (A=a) } =
-    \frac{1}{ P_{A} (A=a) } \big( \EE [ \mathcal{S}(O) \mid A=a] -
    \EE [ \mathcal{S}(O)]\big). 
\end{align*}
We now split up the functional derivative of \(\Psi_{a^*}\) as
follows:
\begin{align}
  &\frac{d}{d\eps} \Psi_{a^*}(P_{\eps} )
  =  \frac{d}{d\eps}
    \int_{\mathcal{W}} \sum_{z=0,1}  \int_{\mathcal{Y}} y \, dP_{Y,\eps} (y  \mid z, 1, w) 
    \gamma_{\eps}( z \mid a^*, w) 
    \frac{\pi_{\eps}( 1 \mid w)}{P_{\eps}(A=1)}  d\mu_{\eps} (w)\notag \\
  & \quad =  \frac{d}{d\eps}
    \int_{\mathcal{W}} \sum_{a=0,1} \sum_{z=0,1}  \int_{\mathcal{Y}} y \, dP_{Y,\eps} (y  \mid z, a, w) 
    \frac{\gamma( z \mid a^*, w) }{\gamma( z \mid a, w)}
    \frac{\1\lbrace a=1\rbrace}{\bar{\pi}(1)}  \gamma(z \mid a, w)\pi( a \mid w)  d\mu (w) \label{eq:term:y} \\
  & \quad +  \frac{d}{d\eps}
    \int_{\mathcal{W}} \sum_{a=0,1}  \sum_{z=0,1}  \int_{\mathcal{Y}} y \, dP_{Y} (y  \mid z, 1, w) 
    \gamma_{\eps}( z \mid a, w) \frac{\1\lbrace a=a^*\rbrace}{\pi(a \mid w) } \pi(a \mid w)
    \frac{\pi( 1 \mid w)}{\bar{\pi}(1)}  d\mu (w) \label{eq:term:z} \\
  & \quad +  \frac{d}{d\eps}
    \int_{\mathcal{W}} \sum_{a=0,1}\sum_{z=0,1}  \int_{\mathcal{Y}} y \, dP_{Y} (y  \mid z, 1, w) 
    \gamma( z \mid a^*, w) 
    \frac{\pi_{\eps}( a \mid w)}{\bar{\pi}(1)} \1\lbrace a=1\rbrace d\mu (w) \label{eq:term:a} \\
  & \quad +  \frac{d}{d\eps}
    \int_{\mathcal{W}} \sum_{a=0,1} \sum_{z=0,1}  \int_{\mathcal{Y}} y \, dP_{Y} (y  \mid z, 1, w) 
    \gamma( z \mid a^*, w) 
    \frac{\pi( 1 \mid w)}{P_{\eps}(A=a)} \1\lbrace a=1\rbrace  d\mu (w) \label{eq:term:a:marginal} \\
  & \quad +  \frac{d}{d\eps}
    \int_{\mathcal{W}} \sum_{z=0,1}  \int_{\mathcal{Y}} y \, dP_{Y} (y  \mid z, 1, w) 
    \gamma( z \mid a^*, w) 
    \frac{\pi( 1 \mid w)}{P(A=0)}  d\mu_{\eps} (w) . \label{eq:term:w}
\end{align}
We consider each of the terms \eqref{eq:term:y}--\eqref{eq:term:w} one
by one. First, look at \eqref{eq:term:y}:
\begin{align*}
  & \frac{d}{d\eps}\bigg\vert_{\eps = 0}
    \int_{\mathcal{W}} \sum_{z=0,1}  \int_{\mathcal{Y}} y \, dP_{Y,\eps} (y  \mid z, a, w) 
    \frac{\gamma( z \mid a^*, w) }{\gamma( z \mid a, w)} 
    \frac{\1\lbrace a=1\rbrace}{\bar{\pi}(1)}  \gamma(z \mid a, w) \pi( a \mid w)  d\mu (w)  \\
  & \quad =    \int_{\mathcal{W}} \sum_{z=0,1}  \int_{\mathcal{Y}} y 
    \Big( \EE \big[ \mathcal{S}(O) \mid Y=y, Z=z, A=a, W=w\big] - \EE \big[ \mathcal{S}(O) \mid Z=z, A=a, W=w\big] \Big) 
   \\[-0.3cm]
&\qquad\qquad\qquad\qquad\qquad\qquad\qquad\qquad     dP_{Y} (y  \mid z, a, w)\frac{\gamma( z \mid a^*, w) }{\gamma( z \mid a, w)}
    \frac{\1\lbrace a=1\rbrace}{\bar{\pi}(1)}  \gamma(z \mid a, w)\pi( a \mid w)  d\mu (w)  \\
  &= \EE \bigg[  \mathcal{S}(O) Y  \frac{\gamma( Z \mid a^*, W) }{\gamma( Z \mid A, W)}
    \frac{\1\lbrace A=1\rbrace}{\bar{\pi}(1)} \bigg] - \EE \bigg[ \mathcal{S}(O) Q(Z, A, W)
    \frac{\gamma( Z \mid a^*, W) }{\gamma( Z \mid A, W)}
    \frac{\1\lbrace A=1\rbrace}{\bar{\pi}(1)} \bigg], 
\end{align*}
and we see that
\begin{align*}
  \phi_{Y, a^*}^*(P) (O)
  &=  \frac{\gamma( Z \mid a^*, W) }{\gamma( Z \mid A, W)}
  \frac{\1\lbrace A=1\rbrace}{\bar{\pi}(1)}
  \Big( Y - Q(Z, A, W)\Big) \\
  &=
  \begin{cases}
    \frac{\1\lbrace A=1\rbrace}{\bar{\pi}(1)}
    \Big( Y - Q(Z, A, W)\Big) & a^*=1, \\
    \frac{\gamma( Z \mid 0, W) }{\gamma( Z \mid 1, W)}
    \frac{\1\lbrace A=0\rbrace}{P(A=0)}
    \Big( Y - Q(Z, A, W)\Big) & a^*=0. 
    \end{cases}
\end{align*}
Second, consider \eqref{eq:term:z}:
\begin{align*}
  & \frac{d}{d\eps}\bigg\vert_{\eps = 0}    \int_{\mathcal{W}} \sum_{a=0,1}
    \sum_{z=0,1}  \int_{\mathcal{Y}} y \, dP_{Y} (y  \mid z, 1, w) 
    \gamma_{\eps}( z \mid a, w) \frac{\1\lbrace a=a^*\rbrace}{\pi(a \mid w) }
    \frac{\pi( 1 \mid w)}{\bar{\pi}(1)}   \pi(a \mid w) d\mu (w)  \\
  &\quad  =    \int_{\mathcal{W}} \sum_{a=0,1}  \sum_{z=0,1}  \int_{\mathcal{Y}} y \, dP_{Y} (y  \mid z, 1, w) 
\Big( \EE \big[ \mathcal{S}(O) \mid Z=z, A=a, W=w \big] - \EE \big[ \mathcal{S}(O) \mid A=a, W=w \big] \Big)
   \\[-0.13cm]
  &\qquad\qquad\qquad\qquad\qquad\qquad\qquad\qquad\qquad\qquad\qquad\qquad  \gamma( z \mid a, w) 
    \frac{\1\lbrace a=a^*\rbrace}{\pi(a \mid w) }
    \frac{\pi( 1 \mid w)}{\bar{\pi}(1)}   \pi(a \mid w) d\mu (w)  \\
  & \quad = 
    \EE \bigg[ \mathcal{S}(O) Q(Z, 1, W)    \frac{\1\lbrace A=a^*\rbrace}{\pi(A \mid W) }
    \frac{\pi( 1 \mid W)}{\bar{\pi}(1)}  \bigg] \\
    &\qquad\qquad\qquad\qquad\qquad\qquad\qquad\quad -
    \EE \bigg[ \mathcal{S}(O)    \frac{\1\lbrace A=a^*\rbrace}{\pi(A \mid W) } \frac{\pi( 1 \mid W)}{\bar{\pi}(1)}
    \sum_{z=0,1} Q(z, 0, W) \gamma(z \mid A , W) \bigg], 
\end{align*}
from which we derive that 
\begin{align*}
  \phi_{\gamma, a^*}^*(P) (O)
  &=     \frac{\1\lbrace A=a^*\rbrace}{\pi(A \mid W) } \frac{\pi( 1 \mid W)}{\bar{\pi}(1)}
    \bigg( Q(Z, 1, W) - \sum_{z=0,1} Q(z, 1, W) \gamma(z \mid A , W) \bigg)  \\
  &= \begin{cases}
    \frac{\1\lbrace A=1 \rbrace}{\bar{\pi}(1)}
    \bigg( Q(Z, 1, W) - \sum_{z=0,1} Q(z, 1, W) \gamma(z \mid A , W) \bigg)  , & a^*=1, \\
    \frac{\1\lbrace A=0\rbrace}{\pi(0 \mid W) } \frac{\pi( 1 \mid W)}{\bar{\pi}(1)}
    \bigg( Q(Z, 1, W) - \sum_{z=0,1} Q(z, 1, W) \gamma(z \mid A , W) \bigg)  , & a^*=1. 
  \end{cases}
\end{align*}
Third, consider \eqref{eq:term:a}:
\begin{align*}
  & \frac{d}{d\eps}\bigg\vert_{\eps = 0}
    \int_{\mathcal{W}} \sum_{a=0,1}\sum_{z=0,1}  \int_{\mathcal{Y}} y \, dP_{Y} (y  \mid z, 1, w) 
    \gamma( z \mid a^*, w) 
    \frac{\1\lbrace a=1\rbrace }{\bar{\pi}(1)} \pi_{\eps}( a \mid w)  d\mu (w) \\
  & \quad  =   \int_{\mathcal{W}} \sum_{a=0,1}\sum_{z=0,1}  \int_{\mathcal{Y}} y \, dP_{Y} (y  \mid z, 1, w) 
    \gamma( z \mid a^*, w) 
    \frac{\1\lbrace a=1\rbrace }{\bar{\pi}(1)}
    \Big( \EE \big[ \mathcal{S}(O) \mid A=a, W=w \big] \\[-0.3cm]
  &\qquad\qquad\qquad\qquad\qquad\qquad\qquad\qquad\qquad\qquad\qquad\qquad\qquad\qquad
    - \EE \big[ \mathcal{S}(O) \mid W= w \big] \Big)
    \pi( a \mid w)    d\mu (w)  \\
  & \quad = \EE \bigg[ \mathcal{S}(O)  \frac{\1\lbrace A=1\rbrace }{\bar{\pi}(1)}
    \sum_{z=0,1} Q(z, 1, W) \gamma( z \mid a^*, W) \bigg]
    - \EE \bigg[ \mathcal{S}(O)  \frac{\pi ( 1 \mid W) }{\bar{\pi}(1)}
    \sum_{z=0,1} Q(z, 1, W) \gamma( z \mid a^*, W) \bigg], 
\end{align*}
from which we see that
\begin{align*}
  \phi_{\pi, a^*}^*(P) (O)
  &=   \bigg( \frac{\1\lbrace A=1\rbrace }{\bar{\pi}(1)} - \frac{\pi ( 1 \mid W) }{\bar{\pi}(1)} \bigg)  
    \sum_{z=0,1} Q(z, 1, W) \gamma( z \mid a^*, W), \qquad a^*=0,1. 
\end{align*}
Next, consider \eqref{eq:term:a:marginal}:
\begin{align*}
  & \frac{d}{d\eps}\bigg\vert_{\eps = 0}
    \int_{\mathcal{W}} \sum_{a=0,1} \sum_{z=0,1}  \int_{\mathcal{Y}} y \, dP_{Y} (y  \mid z,1, w) 
    \gamma( z \mid a^*, w) 
    \frac{\1\lbrace a=1\rbrace}{P_{\eps}(A=a)}  \pi( a \mid w)  d\mu (w) \\
  & \quad = - \int_{\mathcal{W}} \sum_{a=0,1} \sum_{z=0,1}  \int_{\mathcal{Y}} y \, dP_{Y} (y  \mid z, 1, w) 
    \gamma( z \mid a^*, w) \1\lbrace a=1\rbrace
    \bigg(    \frac{ \EE [\mathcal{S}(O) \mid A=a] - \EE [ \mathcal{S}(O)]
    }{P(A=a)} \bigg) \pi( a \mid w)  d\mu (w)\\
  & \quad = - \bigg( \EE \bigg[
    \EE \big[ \mathcal{S}(O)  \mid A \big]   \frac{\1\lbrace A=1 \rbrace}{\bar{\pi}(1)}
    \EE \bigg[ \sum_{z=0,1} Q(z, 1, W) \gamma(z \mid a^*, W) \, \bigg\vert \, A\bigg] \bigg]
    - \EE \big[ \mathcal{S}(O) \Psi_{a^*}(P) \big] \bigg)\\
  & \quad = - \bigg( \EE \bigg[
    \EE \big[ \mathcal{S}(O)  \mid A \big]   \frac{\1\lbrace A=1 \rbrace}{\bar{\pi}(1)}
    \EE \bigg[ \sum_{z=0,1} Q(z, 1, W) \gamma(z \mid a^*, W) \, \bigg\vert \, A\bigg] \bigg]
    - \EE \big[ \mathcal{S}(O) \Psi_{a^*}(P) \big] \bigg)\\
  & \quad = - \bigg( \EE \bigg[
    \EE \big[ \mathcal{S}(O)  \mid A \big]   \frac{\1\lbrace A=1 \rbrace}{\bar{\pi}(1)}
    \EE \bigg[ \sum_{z=0,1} Q(z, 1, W) \gamma(z \mid a^*, W) \, \bigg\vert \, A=1\bigg] \bigg]
    - \EE \big[ \mathcal{S}(O) \Psi_{a^*}(P) \big] \bigg)\\
  & \quad = - \bigg( \EE \bigg[ \mathcal{S}(O) \frac{\1\lbrace A=1 \rbrace}{\bar{\pi}(1)} \Psi_{a^*}(P)\bigg]
    - \EE \big[ \mathcal{S}(O) \Psi_{a^*}(P) \big] \bigg) , 
\end{align*}
i.e., 
\begin{align*}
  \phi_{\bar{\pi}, a^*}^*(P) (O)
  &=  - \bigg( \frac{\1\lbrace A=1 \rbrace}{\bar{\pi}(1)}  - 1\bigg) \Psi_{a^*}(P),  \qquad a^*=0,1. 
\end{align*}
Lastly, consider \eqref{eq:term:w}:
\begin{align*}
  & \frac{d}{d\eps}\bigg\vert_{\eps = 0}
    \int_{\mathcal{W}} \sum_{a=0,1} \sum_{z=0,1}  \int_{\mathcal{Y}} y \, dP_{Y} (y  \mid z, 1, w) 
    \gamma( z \mid a^*, w) 
    \frac{\pi( 1 \mid w)  }{\bar{\pi}(1)} \1 \lbrace a= 1\rbrace d\mu_{\eps} (w)  \\
  & \quad =   \int_{\mathcal{W}} \sum_{a=0,1}  \sum_{z=0,1}  \int_{\mathcal{Y}} y \, dP_{Y} (y  \mid z, 1, w) 
    \gamma( z \mid a^*, w) 
    \frac{\1 \lbrace a= 1\rbrace  }{\bar{\pi}(1)}  \Big( \EE [\mathcal{S}(O) \mid W=w] - \EE [\mathcal{S}(O) ] \Big)
    \pi( a \mid w)d\mu (w)  \\
  & \quad = \EE \bigg[ \mathcal{S}(O)  \frac{\pi(1 \mid W)}{\bar{\pi}(1)}
    \sum_{z=0,1} Q(z, 1, W) \gamma(z \mid a^*, W) \bigg] - \EE \big[ \mathcal{S}(O) \Psi_{a^*}(P) \big]
.
\end{align*}
If we collect the last three expressions,
\begin{align*}
  \phi_{\pi, a^*}^*(P) (O)&=   \bigg( \frac{\1\lbrace A=1\rbrace }{\bar{\pi}(1)} - \frac{\pi ( 1 \mid W) }{\bar{\pi}(1)} \bigg)  
    \sum_{z=0,1} Q(z, 1, W) \gamma( z \mid a^*, W) \\
  \phi_{\bar{\pi}, a^*}^*(P) (O)
  &=  - \bigg( \frac{\1\lbrace A=1 \rbrace}{\bar{\pi}(1)}  - 1\bigg) \Psi_{a^*}(P) \\
  \phi_{\mu, a^*}^*(P) (O)
  &=  \frac{\pi(1 \mid W)}{\bar{\pi}(1)}
    \sum_{z=0,1} Q(z, 1, W) \gamma(z \mid a^*, W) - \Psi_{a^*}(P),
\end{align*}
we immediately we see that some terms cancel out, i.e., 
\begin{align*}
  \phi_{\pi, a^*}^*(P) (O) + \phi_{\bar{\pi}, a^*}^*(P) (O) + 
  \phi_{\mu, a^*}^*(P) (O) =  \frac{\1\lbrace A=1\rbrace }{\bar{\pi}(1)} \bigg(\sum_{z=0,1} Q(z, 1, W)
  \gamma(z \mid a^* , W)
  - \Psi_{a^*} (P) \bigg). 
\end{align*}
Thus, we conclude that the efficient influence function for the
parameter \(\Psi_{a^*}(P)\) is: 
\begin{align*}
  \phi_{a^*} (P) (O)& =  \phi_{Y, a^*}^*(P) (O)
    +  \phi_{\gamma, a^*}^*(P) (O) +
    \phi_{\pi, a^*}^*(P) (O) +
     \phi_{\bar{\pi}, a^*}^*(P) (O) + 
    \phi_{\mu, a^*}^*(P) (O) \\
  &  = \frac{\gamma( Z \mid a^*, W) }{\gamma( Z \mid 1, W)}
    \frac{\1\lbrace A=1\rbrace}{\bar{\pi}(1)}
    \Big( Y - Q(Z, A, W)\Big) \\
  &\qquad\quad  + \frac{\1\lbrace A=a^*\rbrace}{\pi(A \mid W) } \frac{\pi( 1 \mid W)}{\bar{\pi}(1)}
    \bigg( Q(Z, 1, W) - \sum_{z=0,1} Q(z, 1, W) \gamma(z \mid A , W) \bigg) \\
  &\qquad\quad + \frac{\1\lbrace A=1\rbrace }{\bar{\pi}(1)} \bigg(\sum_{z=0,1} Q(z, 1, W) \gamma(z \mid a^* , W)
    - \Psi_{a^*} (P) \bigg)  .
\end{align*}

\subsection{Event history setting}
\label{app:eff:ic:survival}

As in Section \ref{app:eff:ic:binary}, we derive the efficient
influence functions for the two parameters \(\Psi_0(P)\) and
\(\Psi_1(P)\) separately. Fix \(a^*=0,1\) and consider the parameter
\(\Psi_{a^*} \, :\, \mathcal{M}\rightarrow\R\), i.e.,
\begin{align*}
    \Psi_{a^*}(P)
    &= \EE \bigg[ \sum_{z=0,1} F_1 (\tau  \mid  Z=z, A, W)   \gamma(z \mid a^*, W)  \bigg\vert A=1\bigg] \\
    &=   \int_{\mathcal{W}} \sum_{z=0,1}  F_1 (\tau  \mid  Z=z, A, W) 
      \gamma( z \mid a^*, w) 
      \frac{\pi( 1 \mid w)}{\bar{\pi}(1)}  d\mu (w) \\
    &=   \int_{\mathcal{W}} \sum_{z=0,1} \bigg( \int_0^\tau \lambda_1 (t \mid z, 1, w) S(t-\mid z,1,w)dt\bigg) 
      \gamma( z \mid a^*, w) 
      \frac{\pi( 1 \mid w)}{\bar{\pi}(1)}  d\mu (w).
\end{align*}
The observed data density now factorizes as follows
\begin{align*}
  p( o \mid A=1) = p_{\tilde{T},\tilde{\Delta}} (t, \delta \mid z,a,w) \gamma (z \mid a, w)\frac{\pi (a \mid w)}{\bar{\pi}(1)} \mu (w),  
\end{align*}
where
\begin{align*}
 p_{\tilde{T},\tilde{\Delta}} (t, \delta \mid z,a,w) = \prod_{j=1}^J \big( \lambda_j ( t \mid z, a, w) \big)^{\1\lbrace \delta=j\rbrace}
  S(t- \mid z,a,w) \big( \lambda^c(t \mid z, a, w) \big)^{\1\lbrace \delta=0\rbrace}
  S^c(t-\mid z, a,w).
\end{align*}
Again, the score \(\mathcal{S}\) decomposes into an orthogonal sum of
factor-specific scores,
\begin{align*}
  \mathcal{S} =
  \underbrace{\frac{d}{d\eps}\bigg\vert_{\eps =0 } \log p_{T,\Delta,\eps}}_{=\mathcal{S}_{T,\Delta}}
  +  \underbrace{ \frac{d}{d\eps}\bigg\vert_{\eps =0 } \log \gamma_{\eps}}_{=\mathcal{S}_{\gamma}}
  +  \underbrace{ \frac{d}{d\eps}\bigg\vert_{\eps =0 } \log \pi_{\eps}}_{=\mathcal{S}_{\pi}} 
  -  \underbrace{ \frac{d}{d\eps}\bigg\vert_{\eps =0 } \log P_{A,\eps}}_{=\mathcal{S}_{\bar{\pi}}}
  +  \underbrace{ \frac{d}{d\eps}\bigg\vert_{\eps =0 } \log \mu_{\eps}}_{=\mathcal{S}_{\mu}}  .
\end{align*}
We again split up the functional derivative of \(\Psi_{a^*}\) as
follows:
\begin{align}
  &\frac{d}{d\eps} \Psi_{a^*}(P_{\eps} )
    =     \frac{d}{d\eps}\bigg\vert_{\eps = 0}
    \int_{\mathcal{W}} \sum_{z=0,1} \bigg( \int_0^{\tau} {\lambda}_{1,\eps}(t \mid z, a, w)
    {S}_{\eps} (t- \mid z, a, w)  dt\bigg) \gamma_{\eps}( z \mid a^*, w) 
    \frac{\pi_{\eps}( 1 \mid w)}{P_{\eps}(A=1)}  d\mu_{\eps} (w)\notag \\
  & \quad =  \frac{d}{d\eps}
    \int_{\mathcal{W}} \sum_{a=0,1} \sum_{z=0,1}  \bigg( \int_0^{\tau} {\lambda}_{1,\eps}(t \mid z, a, w)
    {S}_{\eps} (t- \mid z, a, w)  dt\bigg) 
    \frac{\gamma( z \mid a^*, w) }{\gamma( z \mid a, w)}
    \frac{\1\lbrace a=1\rbrace}{\bar{\pi}(1)}  \gamma(z \mid a, w)\pi( a \mid w)  d\mu (w) \label{eq:term:survival} \\
  & \quad +  \frac{d}{d\eps}
    \int_{\mathcal{W}} \sum_{a=0,1}  \sum_{z=0,1}  \bigg( \int_0^{\tau} {\lambda}_{1,\eps}(t \mid z, 1, w)
    {S}_{\eps} (t- \mid z, 1, w)  dt\bigg) 
    \gamma_{\eps}( z \mid a, w) \frac{\1\lbrace a=a^*\rbrace}{\pi(a \mid w) } \pi(a \mid w)
    \frac{\pi( 1 \mid w)}{\bar{\pi}(1)}  d\mu (w) \notag \\
  & \quad +  \frac{d}{d\eps}
    \int_{\mathcal{W}} \sum_{a=0,1}\sum_{z=0,1}  \bigg( \int_0^{\tau} {\lambda}_{1,\eps}(t \mid z, 1, w)
    {S}_{\eps} (t- \mid z, 1, w)  dt\bigg) 
    \gamma( z \mid a^*, w) 
    \frac{\pi_{\eps}( a \mid w)}{\bar{\pi}(1)} \1\lbrace a=1\rbrace d\mu (w) \notag \\
  & \quad +  \frac{d}{d\eps}
    \int_{\mathcal{W}} \sum_{a=0,1} \sum_{z=0,1}   \bigg( \int_0^{\tau} {\lambda}_{1,\eps}(t \mid z, 1, w)
    {S}_{\eps} (t- \mid z, 1, w)  dt\bigg) 
    \gamma( z \mid a^*, w) 
    \frac{\pi( 1 \mid w)}{P_{\eps}(A=a)} \1\lbrace a=1\rbrace  d\mu (w) \notag \\
  & \quad +  \frac{d}{d\eps}
    \int_{\mathcal{W}} \sum_{z=0,1}  \bigg( \int_0^{\tau} {\lambda}_{1,\eps}(t \mid z, 1, w)
    {S}_{\eps} (t- \mid z, 1, w)  dt\bigg) 
    \gamma( z \mid a^*, w) 
    \frac{\pi( 1 \mid w)}{P(A=0)}  d\mu_{\eps} (w) ; \notag
\end{align}
the only term for which the calculations differ from Section
\ref{app:eff:ic:binary} is that corresponding to
\eqref{eq:term:survival}. The computation of this term follows exactly
as in \cite[][Appendix V]{rytgaard2021estimation}, and we get that
\begin{align*}
  &\phi_{\tilde{T},\tilde{\Delta}, a^*}^*(P) (O) =
  \frac{\gamma( Z \mid a^*, W) }{\gamma( Z \mid A, W)}
  \frac{\1\lbrace A=1\rbrace}{\bar{\pi}(1)}
\bigg(  \int_0^{\tau}
    \big( S^c(t- \mid A,W)\big)^{-1} \\
  &\qquad\qquad \times \bigg( 1-
    \frac{F_1 (\tau \mid A,W) - F_1(t \mid A,W)}{S(t \mid A,W)}\bigg)\big(N_1 (dt) - \1\lbrace
    \tilde{T}\ge t \rbrace \lambda_1 (t \mid A,W) dt \big) \\
 &\qquad\qquad -\sum_{l\neq 1} \int_0^{\tau}
    \big( S^c(t- \mid A,W)\big)^{-1} \bigg( 
   \frac{F_1 (\tau \mid A,W)- F_1(t \mid A,W)}{S(t \mid A,W)}\bigg) \\[-0.3cm]
    &\qquad\qquad\qquad\qquad\qquad\qquad\qquad\qquad\qquad\qquad\qquad \times \big(N_l (dt) - \1\lbrace
  \tilde{T}\ge t \rbrace \lambda_l (t \mid A,W) dt \big) \bigg).
\end{align*}
For the remaining terms, \(F_1(\tau \mid z, a, w)\) is substituted for
\(Q(z,a,w)\).

\section{Second-order remainders}
\label{app:remainder}

\subsection{Binary outcome setting}
\label{app:remainder:binary}

We here compute
\(R_{a^*}(P, P_0) = \Psi_{a^*}( P) - \Psi_{a^*}(P_0) + P_0 \phi_{a^*}(P)\). For this purpose, recall that
\begin{align*}
  \Psi_{a^*}(P)
  &= \EE \bigg[ \sum_{z=0,1} \EE[ Y \mid  Z=z, A, W]   \gamma(z \mid a^*, W)  \bigg\vert A=1\bigg] \\
  &=   \int_{\mathcal{W}} \sum_{z=0,1}  Q (z, 1, w) 
    \gamma( z \mid a^*, w) 
    \frac{\pi( 1 \mid w)}{\bar{\pi}(1)}  d\mu (w) \\
  &=   \int_{\mathcal{W}} \sum_{z=0,1}  \int_{\mathcal{Y}} y \, dP_{Y} (y  \mid z, 1, w) 
    \gamma( z \mid a^*, w) 
    \frac{\pi( 1 \mid w)}{\bar{\pi}(1)}  d\mu (w), 
\end{align*}
and that
\begin{align*}
  \phi_{a^*} (P) (O)
  &  = \frac{\gamma( Z \mid a^*, W) }{\gamma( Z \mid 1, W)}
    \frac{\1\lbrace A=1\rbrace}{\bar{\pi}(1)}
    \Big( Y - Q(Z, A, W)\Big) \\
  &\qquad\quad  + \frac{\1\lbrace A=a^*\rbrace}{\pi(A \mid W) } \frac{\pi( 1 \mid W)}{\bar{\pi}(1)}
    \bigg( Q(Z, 1, W) - \sum_{z=0,1} Q(z, 1, W) \gamma(z \mid A , W) \bigg) \\
  &\qquad\quad + \frac{\1\lbrace A=1\rbrace }{\bar{\pi}(1)} \bigg(\sum_{z=0,1} Q(z, 1, W) \gamma(z \mid a^* , W)
    - \Psi_{a^*} (P) \bigg)  .
\end{align*}
Now, first note that
\begin{align*}
  &  \EE_{P_0}
    \bigg[ \frac{\gamma( Z \mid a^*, W) }{\gamma( Z \mid 1, W)}
    \frac{\1\lbrace A=1\rbrace}{\bar{\pi}(1)}
    \Big( Y - Q(Z, A, W)\Big)\bigg] \\
  & \qquad =  \EE_{P_0}  \bigg[ \EE_{P_0}
    \bigg[ \frac{\gamma( Z \mid a^*, W) }{\gamma( Z \mid 1, W)}
    \frac{\1\lbrace A=1\rbrace}{\bar{\pi}(1)}
    \Big( Y - Q(Z, A, W)\Big) \,\bigg\vert \, Z,A,W \bigg] \bigg]\\
  & \qquad =  \EE_{P_0}  \bigg[ \frac{\gamma( Z \mid a^*, W) }{\gamma( Z \mid 1, W)}
    \frac{\1\lbrace A=1\rbrace}{\bar{\pi}(1)}
    \Big( Q_0(Z,1,W) - Q(Z, 1, W)\Big)  \bigg] \\
    & \qquad = \EE_{P_0}  \bigg[ \EE_{P_0}  \bigg[ \frac{\gamma( Z \mid a^*, W) }{\gamma( Z \mid 1, W)}
    \frac{\1\lbrace A=1\rbrace}{\bar{\pi}(1)}
      \Big( Q_0(Z,1,W) - Q(Z, 1, W)\Big)  \,\bigg\vert \, A,W \bigg] \bigg]\\
      & \qquad = \EE_{P_0}   \bigg[ 
    \frac{\1\lbrace A=1\rbrace}{\bar{\pi}(1)}
        \bigg( \sum_{z=0,1} \frac{\gamma( z \mid a^*, W) }{\gamma( z \mid 1, W)}
        \Big( Q_0(z,1,W) - Q(z, 1, W) \Big)\gamma_0 (z \mid 1,W) \bigg)   \bigg]\\
    & \qquad = \EE_{P_0}  \bigg[  \EE_{P_0}  \bigg[  
    \frac{\1\lbrace A=1\rbrace }{\bar{\pi} (1)}
      \bigg( \sum_{z=0,1} \frac{\gamma( z \mid a^*, W) }{\gamma( z \mid 1, W)}
      \Big( Q_0(z,1,W) - Q(z, 1, W) \Big)\gamma_0 (z \mid 1,W) \bigg)
      \,\bigg\vert \, W \bigg] \bigg]\\
    & \qquad = \EE_{P_0}  \bigg[ 
    \frac{\pi_0(1 \mid W) }{\bar{\pi} (1)}
      \bigg( \sum_{z=0,1} \frac{\gamma( z \mid a^*, W) }{\gamma( z \mid 1, W)}
      \Big( Q_0(z,1,W) - Q(z, 1, W) \Big)\gamma_0 (z \mid 1,W) \bigg)
   \bigg] , 
\end{align*}
second that
\begin{align*}
  &  \EE_{P_0}
    \bigg[ \frac{\1\lbrace A=a^*\rbrace}{\pi(A \mid W) } \frac{\pi( 1 \mid W)}{\bar{\pi}(1)}
    \bigg( Q(Z, 1, W) - \sum_{z=0,1} Q(z, 1, W) \gamma(z \mid A , W) \bigg) \bigg] \\
  & \qquad =  \EE_{P_0}  \bigg[ \EE_{P_0}
    \bigg[ \frac{\1\lbrace A=a^*\rbrace}{\pi(A \mid W) } \frac{\pi( 1 \mid W)}{\bar{\pi}(1)}
    \bigg( Q(Z, 1, W) - \sum_{z=0,1} Q(z, 1, W) \gamma(z \mid A , W) \bigg)
    \,\bigg\vert \, A,W \bigg] \bigg]\\
  & \qquad =  \EE_{P_0}  \bigg[ \frac{\1\lbrace A=a^*\rbrace}{\pi( a^* \mid W) }
    \frac{\pi( 1 \mid W)}{\bar{\pi}(1)}
    \sum_{z=0,1} Q(z, 1, W) \big( \gamma_0(z \mid a^* , W)-  \gamma(z \mid a^* , W) \big)  \bigg]\\
  & \qquad =   \EE_{P_0}  \bigg[\EE_{P_0}  \bigg[ \frac{\1\lbrace A=a^*\rbrace}{\pi(a^* \mid W) }
    \frac{\pi( 1 \mid W)}{\bar{\pi}(1)}
    \sum_{z=0,1} Q(z, 1, W) \big( \gamma_0(z \mid a^* , W)-  \gamma(z \mid a^* , W) \big)
    \,\bigg\vert \, W \bigg] \bigg] \\
  & \qquad =   \EE_{P_0}  \bigg[ \frac{\pi_0(a^* \mid W)}{\pi(a^* \mid W) }
    \frac{\pi( 1 \mid W)}{\bar{\pi}(1)}
    \sum_{z=0,1} Q(z, 1, W) \big( \gamma_0(z \mid a^* , W)-  \gamma(z \mid a^* , W) \big)
    \bigg] ,
\end{align*}
and third that 
\begin{align*}
  &  \EE_{P_0}
    \bigg[ \frac{\1\lbrace A=1\rbrace }{\bar{\pi}(1)} \bigg(\sum_{z=0,1} Q(z, 1, W) \gamma(z \mid a^* , W)
    - \Psi_{a^*} (P) \bigg) \bigg] \\
  & \qquad =  \EE_{P_0}  \bigg[ \EE_{P_0}
    \bigg[  \frac{\1\lbrace A=1\rbrace }{\bar{\pi}(1)} \sum_{z=0,1} Q(z, 1, W) \gamma(z \mid a^* , W)
    \,\bigg\vert \, W \bigg] \bigg]
    - \EE_{P_0}  \bigg[ \frac{\1\lbrace A=1\rbrace }{\bar{\pi}(1)} \Psi_{a^*} (P) \bigg]  \\
  & \qquad =  
    \EE_{P_0}  \bigg[ \frac{\pi_0( 1 \mid W)}{ \bar{\pi}(1)}
\sum_{z=0,1} Q(z, 1, W) \gamma(z \mid a^* , W)\bigg]
    -     \frac{\bar{\pi}_0( 1)}{\bar{\pi}(1)} \Psi_{a^*} (P) .
\end{align*}
Collecting the above we have:
\begin{align*}
  R_{a^*}(P,P_0) & = \Psi_{a^*}(P) - \Psi_{a^*}(P_0) \\
           & \,\,
             +\EE_{P_0}  \bigg[ 
             \frac{\pi_0(1 \mid W) }{\bar{\pi} (1)}
             \bigg( \sum_{z=0,1} \frac{\gamma( z \mid a^*, W) }{\gamma( z \mid 1, W)}
             \Big( Q_0(z,1,W) - Q(z, 1, W) \Big)\gamma_0 (z \mid 1,W) \bigg)
             \bigg] \\
           & \,\, + \EE_{P_0}  \bigg[ \frac{\pi_0(a^* \mid W)}{\pi(a^* \mid W) }
             \frac{\pi( 1 \mid W)}{ \bar{\pi}(1) }
             \sum_{z=0,1} Q(z, 1, W) \big( \gamma_0(z \mid a^* , W)-  \gamma(z \mid a^* , W) \big)
             \bigg] \\
           & \,\, +       \EE_{P_0}  \bigg[ \frac{\pi_0( 1 \mid W)}{ \bar{\pi}(1)}
             \sum_{z=0,1} Q(z, 1, W) \gamma(z \mid a^* , W)\bigg]
             -     \frac{\bar{\pi}_0( 1)}{\bar{\pi}(1)} \Psi_{a^*} (P),
             \intertext{to which we add and subtract as follows}
           & \,\,  \pm  \EE_{P_0}  \bigg[ \frac{\pi_0( 1 \mid W)}{ \bar{\pi}(1)}
             \sum_{z=0,1} Q_0(z, 1, W) \gamma_0(z \mid a^* , W)\bigg] .
\end{align*}
Since
\begin{align*}
  &  \EE_{P_0}  \bigg[ \frac{\pi_0( 1 \mid W)}{ \bar{\pi}(1)}
    \sum_{z=0,1} Q_0(z, 1, W) \gamma_0(z \mid a^* , W)\bigg] \\
  &\qquad  = \frac{\bar{\pi}_0(1)}{ \bar{\pi}(1)}
    \EE_{P_0}  \bigg[  \frac{\pi_0( 1 \mid W)}{ \bar{\pi}_0(1)}
    \sum_{z=0,1} Q_0(z, 1, W) \gamma_0(z \mid a^* , W)\bigg] =
    \frac{\bar{\pi}_0(1)}{ \bar{\pi}(1)} \Psi_{a^*}(P_0), 
\end{align*}
we can now write 
\begin{align}
 R_{a^*}(P,P_0)       & = \EE_{P_0}  \bigg[ 
             \frac{\pi_0(1 \mid W) }{\bar{\pi} (1)}
             \bigg( \sum_{z=0,1} \frac{\gamma( z \mid a^*, W) }{\gamma( z \mid 1, W)}
             \big( Q_0(z,1,W) - Q(z, 1, W) \big)\gamma_0 (z \mid 1,W)  \bigg)
                  \bigg] 
  \label{eq:R:1} \\
           & \,\, + \EE_{P_0}  \bigg[ \frac{\pi_0(a^* \mid W)}{\pi(a^* \mid W) }
             \frac{\pi( 1 \mid W)}{ \bar{\pi}(1) }
             \sum_{z=0,1} Q(z, 1, W) \big( \gamma_0(z \mid a^* , W)-  \gamma(z \mid a^* , W) \big)
             \bigg]   \label{eq:R:2} \\
           & \,\, +       \EE_{P_0}  \bigg[ \frac{\pi_0( 1 \mid W)}{ \bar{\pi}(1)}
             \sum_{z=0,1} \Big(  Q(z, 1, W) \gamma(z \mid a^* , W) -
             Q_0(z, 1, W) \gamma_0(z \mid a^* , W) \Big) 
             \bigg]   \label{eq:R:3} \\
           & \,\, +
             \bigg( 1-     \frac{\bar{\pi}_0( 1)}{\bar{\pi}(1)} \bigg) \big( \Psi_{a^*} (P)
             - \Psi_{a^*} (P_0) \big) 
             . \label{eq:R:4}
\end{align}
Now, first consider \eqref{eq:R:1}+\eqref{eq:R:3}, which we shall
write as follows:
\begin{align}
  & \EE_{P_0}  \bigg[ 
    \frac{\pi_0(1 \mid W) }{\bar{\pi} (1)}
    \sum_{z=0,1} \frac{\gamma_0 (z \mid 1,W) }{\gamma( z \mid 1, W)}
    \big( Q_0(z,1,W) - Q(z, 1, W) \big)\gamma( z \mid a^*, W)
    \bigg] \notag\\
  & \qquad\quad +  \EE_{P_0}  \bigg[ \frac{\pi_0( 1 \mid W)}{ \bar{\pi}(1)}
    \sum_{z=0,1} \Big(  Q(z, 1, W) \gamma(z \mid a^* , W) -
    Q_0(z, 1, W) \gamma_0(z \mid a^* , W) \Big) 
    \bigg] \notag \\
  & \qquad\quad \pm \EE_{P_0}  \bigg[ 
    \frac{\pi_0(1 \mid W) }{\bar{\pi} (1)}
    \sum_{z=0,1} 
    \big( Q_0(z,1,W) - Q(z, 1, W) \big)\gamma( z \mid a^*, W) 
    \bigg] \notag \\
  & \qquad = \EE_{P_0}  \bigg[ 
    \frac{\pi_0(1 \mid W) }{\bar{\pi} (1)}
    \sum_{z=0,1} \bigg(  \frac{\gamma_0 (z \mid 1,W) }{\gamma( z \mid 1, W)}  - 1\bigg) 
    \big( Q_0(z,1,W) - Q(z, 1, W) \big)\gamma( z \mid a^*, W)  \bigg)
    \bigg] \label{eq:R:13:1} \\
  & \qquad\quad + \EE_{P_0}  \bigg[ 
    \frac{\pi_0(1 \mid W) }{\bar{\pi} (1)}
    \sum_{z=0,1} 
    Q_0(z,1,W) \big( \gamma( z \mid a^*, W)  -  \gamma_0( z \mid a^*, W) \big)
    \bigg] .  \label{eq:R:13:2}
\end{align}
Next consider \eqref{eq:R:2}+\eqref{eq:R:13:2}, which we shall write
as
\begin{align*}
& \EE_{P_0}  \bigg[ \frac{\pi_0(a^* \mid W)}{\pi(a^* \mid W) }
             \frac{\pi( 1 \mid W)}{ \bar{\pi}(1) }
             \sum_{z=0,1} Q(z, 1, W) \big( \gamma_0(z \mid a^* , W)-  \gamma(z \mid a^* , W) \big)
  \bigg] \\
  & \qquad\quad  \pm  \EE_{P_0}  \bigg[   \frac{\pi_0(1 \mid W) }{\bar{\pi} (1)}
             \sum_{z=0,1} Q(z, 1, W) \big( \gamma(z \mid a^* , W)-  \gamma_0(z \mid a^* , W) \big)
             \bigg] \\
  & \qquad\quad +  \EE_{P_0}  \bigg[ 
    \frac{\pi_0(1 \mid W) }{\bar{\pi} (1)}
    \sum_{z=0,1} 
    Q_0(z,1,W) \big( \gamma( z \mid a^*, W)  -  \gamma_0( z \mid a^*, W) \big)
    \bigg] \\
& \qquad =
  \EE_{P_0}  \bigg[  \bigg( \frac{\pi_0(1 \mid W) }{\bar{\pi} (1)} - \frac{\pi_0(a^* \mid W)}{\pi(a^* \mid W) }
             \frac{\pi( 1 \mid W)}{ \bar{\pi}(1) } \bigg)
             \sum_{z=0,1} Q(z, 1, W) \big( \gamma(z \mid a^* , W)-  \gamma_0(z \mid a^* , W) \big)
  \bigg] \\
   & \qquad\quad +  \EE_{P_0}  \bigg[\frac{\pi_0(1 \mid W) }{\bar{\pi} (1)}
    \sum_{z=0,1} 
    \big( Q_0(z,1,W) - Q(z,1,W)\big)  \big( \gamma( z \mid a^*, W)  -  \gamma_0( z \mid a^*, W) \big)
    \bigg].
\end{align*}
Collecting all terms
\eqref{eq:R:13:1}+\eqref{eq:R:2}+\eqref{eq:R:13:2}+\eqref{eq:R:4} above we see that
\begin{align*}
  & R_{a^*}(P,P_0)
   =
    \EE_{P_0}  \bigg[ 
    \frac{\pi_0(1 \mid W) }{\bar{\pi} (1)}
    \sum_{z=0,1} \bigg(  \frac{\gamma_0 (z \mid 1,W) - \gamma( z \mid 1, W) }{\gamma( z \mid 1, W)}  \bigg) 
    \big( Q_0(z,1,W) - Q(z, 1, W) \big)\gamma( z \mid a^*, W)  \bigg)
    \bigg] \\
  & \qquad +   \EE_{P_0}  \bigg[  \bigg( \frac{\pi_0(1 \mid W) }{\bar{\pi} (1)} - \frac{\pi_0(a^* \mid W)}{\pi(a^* \mid W) }
             \frac{\pi( 1 \mid W)}{ \bar{\pi}(1) } \bigg)
             \sum_{z=0,1} Q(z, 1, W) \big( \gamma(z \mid a^* , W)-  \gamma_0(z \mid a^* , W) \big)
  \bigg] \\
  & \qquad  + \EE_{P_0}  \bigg[\frac{\pi_0(1 \mid W) }{\bar{\pi} (1)}
    \sum_{z=0,1} 
    \big( Q_0(z,1,W) - Q(z,1,W)\big)  \big( \gamma( z \mid a^*, W)  -  \gamma_0( z \mid a^*, W) \big)
    \bigg] \\
  & \qquad + 
             \bigg( 1-     \frac{\bar{\pi}_0( 1)}{\bar{\pi}(1)} \bigg) \big( \Psi_{a^*} (P)
             - \Psi_{a^*} (P_0) \big) . 
\end{align*}
Note that
\begin{align*}
  & \frac{\pi_0(1 \mid W) }{\bar{\pi} (1)} -
    \frac{\pi_0(1 \mid W)}{\pi(1 \mid W) }
    \frac{\pi( 1 \mid W)}{ \bar{\pi}(1) } =  \frac{\pi_0(1 \mid W) }{\bar{\pi} (1)} -
    \frac{\pi_0(1 \mid W)}{ \bar{\pi}(1) } = 0
    \intertext{and,}
  & \frac{\pi_0(1 \mid W) }{\bar{\pi} (1)} -
    \frac{\pi_0(0 \mid W)}{\pi(0 \mid W) }
    \frac{\pi( 1 \mid W)}{ \bar{\pi}(1) }  \\
       & \quad = \frac{\pi_0(1 \mid W)  ( 1- \pi( 1 \mid W))
    -  (1-      \pi_0(1 \mid W)) \pi( 1 \mid W)}{\bar{\pi} (1) ( 1- \pi( 1 \mid W))}
    =
    \frac{\pi_0(1 \mid W)  - \pi( 1 \mid W)
    }{\bar{\pi} (1) ( 1- \pi( 1 \mid W))}
    , 
\end{align*}
as well as, 
\begin{align*}
  \bigg(  \frac{\gamma_0 (z \mid 1,W) - \gamma( z \mid 1, W) }{\gamma( z \mid 1, W)}  \bigg)
  \gamma( z \mid 1, W) = \gamma_0 (z \mid 1,W) - \gamma( z \mid 1, W),
\end{align*}
i.e.,
\begin{align*}
&  R_{1}(P,P_0)    =
                       \bigg( 1-     \frac{\bar{\pi}_0( 1)}{\bar{\pi}(1)} \bigg) \big( \Psi_{1} (P)
                       - \Psi_{1} (P_0) \big) , 
                       \intertext{and,}
                      & R_{0}(P,P_0)  
 =
    \EE_{P_0}  \bigg[ 
    \frac{\pi_0(1 \mid W) }{\bar{\pi} (1)}
    \sum_{z=0,1} \bigg(  \frac{\gamma_0 (z \mid 1,W) - \gamma( z \mid 1, W) }{\gamma( z \mid 1, W)}  \bigg) 
    \big( Q_0(z,1,W) - Q(z, 1, W) \big)\gamma( z \mid 0, W)  \bigg)
    \bigg] \\
  & \qquad +   \EE_{P_0}  \bigg[  \bigg(  \frac{\pi_0(1 \mid W)  - \pi( 1 \mid W)
                }{\bar{\pi} (1) ( 1- \pi( 1 \mid W))} \bigg)
             \sum_{z=0,1} Q(z, 1, W) \big( \gamma(z \mid 0 , W)-  \gamma_0(z \mid 0 , W) \big)
  \bigg] \\
  & \qquad  + \EE_{P_0}  \bigg[\frac{\pi_0(1 \mid W) }{\bar{\pi} (1)}
    \sum_{z=0,1} 
    \big( Q_0(z,1,W) - Q(z,1,W)\big)  \big( \gamma( z \mid 0, W)  -  \gamma_0( z \mid 0, W) \big)
    \bigg] \\
  & \qquad + 
             \bigg( 1-     \frac{\bar{\pi}_0( 1)}{\bar{\pi}(1)} \bigg) \big( \Psi_{0} (P)
             - \Psi_{0} (P_0) \big).
\end{align*}

\subsection{Event history setting}
\label{app:remainder:survival}

We here compute
\(R_{a^*}(P, P_0) = \Psi_{a^*}( P) - \Psi_{a^*}(P_0) + P_0
\phi_{a^*}(P)\) for the event history setting. For this purpose,
recall that
\begin{align*}
  \Psi_{a^*}(P)
  &= \EE \bigg[ \sum_{z=0,1} F_1 (\tau  \mid  Z=z, A, W)   \gamma(z \mid a^*, W)  \bigg\vert A=1\bigg] \\
  &=   \int_{\mathcal{W}} \sum_{z=0,1}  F_1 (\tau  \mid  Z=z, A, W) 
    \gamma( z \mid a^*, w) 
    \frac{\pi( 1 \mid w)}{\bar{\pi}(1)}  d\mu (w) \\
  &=   \int_{\mathcal{W}} \sum_{z=0,1} \bigg( \int_0^\tau \lambda_1 (t \mid z, 1, w) S(t-\mid z,1,w)dt\bigg) 
    \gamma( z \mid a^*, w) 
    \frac{\pi( 1 \mid w)}{\bar{\pi}(1)}  d\mu (w), 
\end{align*}
and that
\begin{align*}
  \phi_{a^*} (P) (O)
  &  = \frac{\gamma( Z \mid a^*, W) }{\gamma( Z \mid 1, W)}
    \frac{\1\lbrace A=1\rbrace}{\bar{\pi}(1)} \
  \bigg(  \int_0^{\tau}
    \big( S^c(t- \mid Z,A,W)\big)^{-1} \\
  &\qquad\qquad \times \bigg( 1-
    \frac{F_1 (\tau \mid Z,A,W) - F_1(t \mid Z,A,W)}{S(t \mid Z,A,W)}\bigg)\big(N_1 (dt) - \1\lbrace
    \tilde{T}\ge t \rbrace \lambda_1 (t \mid Z,A,W) dt \big) \\
 &\qquad\qquad -\sum_{l\neq 1} \int_0^{\tau}
    \big( S^c(t- \mid Z,A,W)\big)^{-1} \bigg( 
   \frac{F_1 (\tau \mid Z,A,W)- F_1(t \mid Z,A,W)}{S(t \mid Z,A,W)}\bigg) \\[-0.3cm]
    &\qquad\qquad\qquad\qquad\qquad\qquad\qquad\qquad\qquad\qquad\qquad \times \big(N_l (dt) - \1\lbrace
  \tilde{T}\ge t \rbrace \lambda_l (t \mid Z,A,W) dt \big) \bigg) \\[-0.1cm]
                    &\qquad\quad  + \frac{\1\lbrace A=a^*\rbrace}{\pi(A \mid W) } \frac{\pi( 1 \mid W)}{\bar{\pi}(1)}
                      \bigg( F_1(\tau \mid Z, 1, W) - \sum_{z=0,1} F_1 (\tau \mid z, 1, W) \gamma(z \mid A , W) \bigg) \\
                    &\qquad\quad + \frac{\1\lbrace A=1\rbrace }{\bar{\pi}(1)} \bigg(\sum_{z=0,1} F_1 (\tau \mid z, 1, W) \gamma(z \mid a^* , W)
                      - \Psi_{a^*} (P) \bigg)  .
\end{align*}
The computations go largely as in Section \ref{app:remainder:binary}.
By iterated expectations, we first write the term corresponding to
\eqref{eq:R:1} as follows
\begin{align*}
  &  \EE_{P_0}
    \bigg[ \frac{\gamma( Z \mid a^*, W) }{\gamma( Z \mid 1, W)}
    \frac{\1\lbrace A=1\rbrace}{\bar{\pi}(1)} 
  \bigg(  \int_0^{\tau}
    \big( S^c(t- \mid Z,A,W)\big)^{-1} \\
  &\qquad\qquad \times \bigg( 1-
    \frac{F_1 (\tau \mid Z,A,W) - F_1(t \mid Z,A,W)}{S(t \mid Z,A,W)}\bigg)\big(N_1 (dt) - \1\lbrace
    \tilde{T}\ge t \rbrace \lambda_1 (t \mid Z,A,W) dt \big) \\
 &\qquad\qquad -\sum_{l\neq 1} \int_0^{\tau}
    \big( S^c(t- \mid Z,A,W)\big)^{-1} \bigg( 
   \frac{F_1 (\tau \mid Z,A,W)- F_1(t \mid Z,A,W)}{S(t \mid Z,A,W)}\bigg) \\[-0.3cm]
    &\qquad\qquad\qquad\qquad\qquad\qquad\qquad\qquad\qquad\qquad\qquad \times \big(N_l (dt) - \1\lbrace
  \tilde{T}\ge t \rbrace \lambda_l (t \mid Z,A,W) dt \big) \bigg)\bigg] \\
  & \qquad =  \EE_{P_0}
    \bigg[\frac{\pi_0(1 \mid W)}{\bar{\pi}(1)} \sum_{z=0,1} \frac{\gamma_0( z \mid 1, W) }{\gamma( z \mid 1, W)}
      \gamma ( z \mid a^*, W)
  \bigg(  \int_0^{\tau}
      \frac{S_0^c(t- \mid z,1,W)}{ S^c(t- \mid z,1,W) } S_0(t- \mid z,1,W) \\
  &\qquad\qquad\qquad \times \bigg( 1-
    \frac{F_1 (\tau \mid z,1,W) - F_1(t \mid z,1,W)}{S(t \mid z,1,W)}\bigg)\big( \Lambda_{0,1} (dt \mid z,1,W)  -  \Lambda_1 (dt \mid z,1,W)  \big) \\
 &\qquad\qquad -\sum_{l\neq 1} \int_0^{\tau}
    \frac{S_0^c(t- \mid z,1,W)}{ S^c(t- \mid z,1,W) } S_0(t- \mid z,1,W)\bigg( 
   \frac{F_1 (\tau \mid z,1,W)- F_1(t \mid z,1,W)}{S(t \mid z,1,W)}\bigg) \\[-0.3cm]
  &\qquad\qquad\qquad\qquad\qquad\qquad\qquad\qquad\qquad\qquad\qquad \times \big(\Lambda_{0,l} (dt \mid z,1,W) -
    \Lambda_l (dt \mid z,1,W)  \big) \bigg)   \bigg].
\end{align*}
The term corresponding to \eqref{eq:R:3} is
\begin{align*}
& \EE_{P_0}
    \bigg[ \frac{\pi_0(1 \mid W)}{\bar{\pi}(1)}  \sum_{z=0,1} \big( F_1(\tau \mid z, 1, W)- F_{0,1}(\tau \mid z, 1, W)
    \big) \gamma_0(z \mid a^* , W) \bigg],
\intertext{to which we add and subtract,}
& \qquad \pm \, \EE_{P_0}
  \bigg[ \frac{\pi_0(1 \mid W)}{\bar{\pi}(1)}  \sum_{z=0,1}
  \big( F_1(\tau \mid z, 1, W)- F_{0,1}(\tau \mid z, 1, W)
    \big) \gamma(z \mid a^* , W) \bigg],
\end{align*}
Following \citet[][Appendix V]{rytgaard2021estimation}, we can rewrite
\( F_{1} (\tau \mid z, 1, W) - F_{0,1} (\tau \mid z, 1, W)\) as follows
\begin{align}
  &  F_1 (\tau \mid z, 1, W) - F_{0,1} (\tau \mid z, 1, W)\label{eq:F1:F0:plug}
  \\
    \begin{split}
  &\quad = \int_0^{\tmax} S_0(t- \mid z, 1, w) 
    \left(1-  \frac{
    F_{1} (\tmax\mid z, 1, w) - F_{1}(t \mid z, 1, w)}{
    S(t \mid z, 1, w)}\right) \big( \Lambda_{0,1}(dt \mid z, 1, w)
    - \Lambda_1( dt \mid z, 1, w)\big)\\
  &\qquad  
   - \int_0^{\tmax} S_0(t- \mid z, 1, w) 
    \left( \frac{
    F_{1} (\tmax\mid z, 1, w) - F_{1}(t \mid z, 1, w)}{
    S(t \mid z, 1, w)}\right) \big( \Lambda_{0,2}(dt \mid z, 1, w)
    - \Lambda_2( dt \mid z, 1, w)\big),
      \end{split}. \notag
\end{align}
Using this, the terms corresponding to \eqref{eq:R:1}+\eqref{eq:R:3}
can be written
\begin{align}
  & \EE_{P_0}
    \bigg[\frac{\pi_0(1 \mid W)}{\bar{\pi}(1)} \sum_{z=0,1} 
    \gamma ( z \mid a^*, W)
    \bigg(  \int_0^{\tau} \bigg( \frac{\gamma_0( z \mid 1, W) }{\gamma( z \mid 1, W)}
    \frac{S_0^c(t- \mid z,1,W)}{ S^c(t- \mid z,1,W) } - 1 \bigg) S_0(t- \mid z,1,W) \notag \\
  &\qquad\qquad\qquad \times \bigg( 1-
    \frac{F_1 (\tau \mid z,1,W) - F_1(t \mid z,1,W)}{S(t \mid z,1,W)}\bigg)\big( \Lambda_{0,1} (dt \mid z,1,W)  -  \Lambda_1 (dt \mid z,1,W)  \big) \notag \\
  &\qquad\qquad -\sum_{l\neq 1} \int_0^{\tau}
    \frac{S_0^c(t- \mid z,1,W)}{ S^c(t- \mid z,1,W) } S_0(t- \mid z,1,W)\bigg( 
    \frac{F_1 (\tau \mid z,1,W)- F_1(t \mid z,1,W)}{S(t \mid z,1,W)}\bigg) \notag \\[-0.3cm]
  &\qquad\qquad\qquad\qquad\qquad\qquad\qquad\qquad\qquad\qquad \times \big(\Lambda_{0,l} (dt \mid z,1,W) -
    \Lambda_l (dt \mid z,1,W)  \big) \bigg)   \bigg] \notag \\
  & \qquad +   \EE_{P_0}
    \bigg[ \frac{\pi_0(1 \mid W)}{\bar{\pi}(1)}  \sum_{z=0,1}
    F_{0,1}(\tau \mid z, 1, W) \big( \gamma(z \mid a^* , W) - \gamma_0(z \mid a^* , W)\big) \bigg]. \label{eq:R:13:S2}
\end{align}
The term corresponding to \eqref{eq:R:2} is
\begin{align*}
&  \EE_{P_0}  \bigg[ \frac{\pi_0(a^* \mid W)}{\pi(a^* \mid W) }
  \frac{\pi( 1 \mid W)}{ \bar{\pi}(1) }
  \sum_{z=0,1} F_{1}(\tau \mid z, 1, W) \big( \gamma_0(z \mid a^* , W)-  \gamma(z \mid a^* , W) \big)
  \bigg]   ,
  \intertext{to which we add and subtract}
  &  \pm  \EE_{P_0}  \bigg[   \frac{\pi_0(1 \mid W) }{\bar{\pi} (1)}
             \sum_{z=0,1} F_{1}(\tau \mid z, 1, W)  \big( \gamma(z \mid a^* , W)-  \gamma_0(z \mid a^* , W) \big)
             \bigg]
  \intertext{which together added to \eqref{eq:R:13:S2} becomes}
  &  \EE_{P_0}  \bigg[ \bigg(  \frac{\pi_0(1 \mid W)}{\bar{\pi}(1)} -
    \frac{\pi_0(a^* \mid W)}{\pi(a^* \mid W) }
  \frac{\pi( 1 \mid W)}{ \bar{\pi}(1) } \bigg) 
  \sum_{z=0,1} F_{1}(\tau \mid z, 1, W) \big( \gamma(z \mid a^* , W)-  \gamma_0(z \mid a^* , W) \big)
    \bigg] \\
& \qquad +
  \EE_{P_0}  \bigg[ \frac{\pi_0(1 \mid W) }{\bar{\pi} (1)} 
  \sum_{z=0,1}  \big( F_{0,1}(\tau \mid z, 1, W)- F_{1}(\tau \mid z, 1, W) \big) \big( \gamma(z \mid a^* , W)-  \gamma_0(z \mid a^* , W) \big)
    \bigg]
\end{align*}
Collecting all terms we now see that
\begin{align*}
  & R_{a^*} (P, P_0)  =
 \EE_{P_0}
    \bigg[\frac{\pi_0(1 \mid W)}{\bar{\pi}(1)} \sum_{z=0,1} 
    \gamma ( z \mid a^*, W)
    \bigg(  \int_0^{\tau} \bigg( \frac{\gamma_0( z \mid 1, W) }{\gamma( z \mid 1, W)}
    \frac{S_0^c(t- \mid z,1,W)}{ S^c(t- \mid z,1,W) } - 1 \bigg) S_0(t- \mid z,1,W) \notag \\
  &\qquad\qquad\qquad \times \bigg( 1-
    \frac{F_1 (\tau \mid z,1,W) - F_1(t \mid z,1,W)}{S(t \mid z,1,W)}\bigg)\big( \Lambda_{0,1} (dt \mid z,1,W)  -  \Lambda_1 (dt \mid z,1,W)  \big) \notag \\
  &\qquad\qquad\qquad\qquad -\sum_{l\neq 1} \int_0^{\tau}
    \frac{S_0^c(t- \mid z,1,W)}{ S^c(t- \mid z,1,W) } S_0(t- \mid z,1,W)\bigg( 
    \frac{F_1 (\tau \mid z,1,W)- F_1(t \mid z,1,W)}{S(t \mid z,1,W)}\bigg) \notag \\[-0.3cm]
  &\qquad\qquad\qquad\qquad\qquad\qquad\qquad\qquad\qquad\qquad \times \big(\Lambda_{0,l} (dt \mid z,1,W) -
    \Lambda_l (dt \mid z,1,W)  \big) \bigg)   \bigg] \\
  & \qquad + \EE_{P_0}  \bigg[ \bigg(  \frac{\pi_0(1 \mid W)}{\bar{\pi}(1)} -
    \frac{\pi_0(a^* \mid W)}{\pi(a^* \mid W) }
  \frac{\pi( 1 \mid W)}{ \bar{\pi}(1) } \bigg) 
  \sum_{z=0,1} F_{1}(\tau \mid z, 1, W) \big( \gamma(z \mid a^* , W)-  \gamma_0(z \mid a^* , W) \big)
    \bigg] \\
& \qquad +
  \EE_{P_0}  \bigg[ \frac{\pi_0(1 \mid W) }{\bar{\pi} (1)} 
  \sum_{z=0,1}  \big( F_{0,1}(\tau \mid z, 1, W)- F_{1}(\tau \mid z, 1, W) \big) \big( \gamma(z \mid a^* , W)-  \gamma_0(z \mid a^* , W) \big)
  \bigg]  \\
    &\qquad +  
                         \bigg( 1-     \frac{\bar{\pi}_0( 1)}{\bar{\pi}(1)} \bigg) \big( \Psi_{a^*} (P)
      - \Psi_{a^*} (P_0) \big), 
\end{align*}
so that, 
\begin{align*}
  &  R_{1}(P,P_0)    =
    \EE_{P_0}
    \bigg[\frac{\pi_0(1 \mid W)}{\bar{\pi}(1)} \sum_{z=0,1} 
    \gamma_0 ( z \mid 1, W)
    \bigg(  \int_0^{\tau} \bigg( \frac{ S_0^c(t- \mid z,1,W) -   S^c(t- \mid z,1,W)}{ S^c(t- \mid z,1,W) }  \bigg)  \notag \\
  &\qquad\qquad \times S_0(t- \mid z,1,W) \bigg( 1-
    \frac{F_1 (\tau \mid z,1,W) - F_1(t \mid z,1,W)}{S(t \mid z,1,W)}\bigg)\big( \Lambda_{0,1} (dt \mid z,1,W)  -  \Lambda_1 (dt \mid z,1,W)  \big) \notag \\
  &\qquad\qquad\qquad\qquad -\sum_{l\neq 1} \int_0^{\tau}
    \frac{S_0^c(t- \mid z,1,W)}{ S^c(t- \mid z,1,W) } S_0(t- \mid z,1,W)\bigg( 
    \frac{F_1 (\tau \mid z,1,W)- F_1(t \mid z,1,W)}{S(t \mid z,1,W)}\bigg) \notag \\[-0.3cm]
  &\qquad\qquad\qquad\qquad\qquad\qquad\qquad\qquad\qquad\qquad \times \big(\Lambda_{0,l} (dt \mid z,1,W) -
    \Lambda_l (dt \mid z,1,W)  \big) \bigg)   \bigg] \\
  &\qquad +  
    \bigg( 1-     \frac{\bar{\pi}_0( 1)}{\bar{\pi}(1)} \bigg) \big( \Psi_{1} (P)
    - \Psi_{1} (P_0) \big) , 
    \intertext{and,}
  & R_{0}(P,P_0)  
    =
    \EE_{P_0}
    \bigg[\frac{\pi_0(1 \mid W)}{\bar{\pi}(1)} \sum_{z=0,1} 
    \gamma ( z \mid 0, W)
    \bigg(  \int_0^{\tau} \bigg( \frac{\gamma_0( z \mid 1, W) S_0^c(t- \mid z,1,W) - \gamma( z \mid 1, W) S^c(t- \mid z,1,W) }{\gamma( z \mid 1, W) S^c(t- \mid z,1,W)}
    \bigg) \notag \\
  &\qquad\qquad \times  S_0(t- \mid z,1,W)\bigg( 1-
    \frac{F_1 (\tau \mid z,1,W) - F_1(t \mid z,1,W)}{S(t \mid z,1,W)}\bigg)\big( \Lambda_{0,1} (dt \mid z,1,W)  -  \Lambda_1 (dt \mid z,1,W)  \big) \notag \\
  &\qquad\qquad\qquad\qquad -\sum_{l\neq 1} \int_0^{\tau}
    \frac{S_0^c(t- \mid z,1,W)}{ S^c(t- \mid z,1,W) } S_0(t- \mid z,1,W)\bigg( 
    \frac{F_1 (\tau \mid z,1,W)- F_1(t \mid z,1,W)}{S(t \mid z,1,W)}\bigg) \notag \\[-0.3cm]
  &\qquad\qquad\qquad\qquad\qquad\qquad\qquad\qquad\qquad\qquad \times \big(\Lambda_{0,l} (dt \mid z,1,W) -
    \Lambda_l (dt \mid z,1,W)  \big) \bigg)   \bigg] \\
  & \qquad + \EE_{P_0}  \bigg[  \bigg(  \frac{\pi_0(1 \mid W)  - \pi( 1 \mid W)
    }{\bar{\pi} (1) ( 1- \pi( 1 \mid W))} \bigg) 
    \sum_{z=0,1} F_{1}(\tau \mid z, 1, W) \big( \gamma(z \mid 0 , W)-  \gamma_0(z \mid 0 , W) \big)
    \bigg] \\
  & \qquad +
    \EE_{P_0}  \bigg[ \frac{\pi_0(1 \mid W) }{\bar{\pi} (1)} 
    \sum_{z=0,1}  \big( F_{0,1}(\tau \mid z, 1, W)- F_{1}(\tau \mid z, 1, W) \big) \big( \gamma(z \mid 0 , W)-  \gamma_0(z \mid 0 , W) \big)
    \bigg]  \\
  &\qquad +  
    \bigg( 1-     \frac{\bar{\pi}_0( 1)}{\bar{\pi}(1)} \bigg) \big( \Psi_{0} (P)
    - \Psi_{0} (P_0) \big). 
\end{align*}
Note that we used for \(R_1(P,P_0)\) that
\begin{align*}
  & \frac{\gamma_0( z \mid 1, W) S_0^c(t- \mid z,1,W) -  \gamma( z \mid 1, W) S^c(t- \mid z,1,W)}{\gamma( z \mid 1, W) S^c(t- \mid z,1,W) }
    \gamma( z \mid 1, W) \\
  & \qquad +
    \frac{ S^c(t- \mid z,1,W) \big( \gamma( z \mid 1, W)  -  \gamma_0( z \mid 1, W) }{\gamma( z \mid 1, W)  S^c(t- \mid z,1,W) } \\ 
  & \qquad\qquad = 
    \frac{\gamma_0( z \mid 1, W) S_0^c(t- \mid z,1,W) -  \gamma_0( z \mid 1, W) S^c(t- \mid z,1,W)}{ S^c(t- \mid z,1,W) }     \\
  & \qquad\qquad = 
    \frac{\gamma_0( z \mid 1, W) \big(  S_0^c(t- \mid z,1,W) -  S^c(t- \mid z,1,W) \big)}{ S^c(t- \mid z,1,W) }    , 
\end{align*}
as well as \eqref{eq:F1:F0:plug}.

\end{document}